\documentclass{imsart}

\newcommand{\Cov}{\mbox{Cov}}

\newcommand{\ind}{\stackrel{\mathrm{ind}}{\sim}}

\usepackage{amsthm}
\usepackage{amsmath}
\usepackage{amssymb, amsfonts}
\usepackage{natbib}
\usepackage{bm}
\usepackage{graphicx}
\usepackage{color}
\usepackage{floatrow}
\usepackage{units}
\usepackage{url}
\usepackage{verbatim}
\usepackage{slashbox}
\usepackage{epstopdf}
\usepackage{booktabs,caption,fixltx2e}
\usepackage[flushleft]{threeparttable}
\usepackage{rotating}
\usepackage{multirow}
\usepackage{physics}
\usepackage{hyperref}
\usepackage{nameref}
\usepackage{placeins} 
\allowdisplaybreaks

\newcommand{\mbf}[1]{\mathbf{#1}}

\newtheorem{theorem}{Theorem}[section]
\newtheorem{lemma}[theorem]{Lemma}

\makeatletter
\let\orgdescriptionlabel\descriptionlabel
\renewcommand*{\descriptionlabel}[1]{%
  \let\orglabel\label
  \let\label\@gobble
  \phantomsection
  \edef\@currentlabel{#1}%
  \let\label\orglabel
  \orgdescriptionlabel{#1}%
}
\makeatother

\begin{document}

\begin{frontmatter}
\title{Bayesian Pairwise Estimation Under Dependent Informative Sampling}
\runtitle{Pairwise Estimation Under Informative Sampling}

\begin{aug}
\author{\fnms{Matthew R.} \snm{Williams}\corref{}\ead[label=e1]{Matthew.Williams@samhsa.hhs.gov}}
\address{Substance Abuse and Mental Health Services Administration\\ 5600 Fishers Ln., Rockville, MD 20857 USA\\ \printead{e1}}

\author{\fnms{Terrance D.} \snm{Savitsky}\ead[label=e2]{Savitsky.Terrance@bls.gov}}
\address{U.S. Bureau of Labor Statistics\\ 2 Massachusetts Ave. N.E, Washington, D.C. 20212 USA\\ \printead{e2}}

\runauthor{M. R. Williams et al.}




\end{aug}

\maketitle
\begin{abstract}
An informative sampling design leads to the selection of units whose inclusion probabilities are correlated with the response variable of interest.  Model inference performed on the resulting observed sample will be biased for the population generative model.  One approach that produces asymptotically unbiased inference employs marginal inclusion probabilities to form sampling weights used to exponentiate each likelihood contribution of a pseudo likelihood used to form a pseudo posterior distribution.  Conditions for posterior consistency restrict applicable sampling designs to those under which pairwise inclusion dependencies asymptotically limit to $0$.  There are many sampling designs excluded by this restriction; for example, a multi-stage design that samples individuals within households.  Viewing each household as a population, the dependence among individuals does not attenuate.  We propose a more targeted approach in this paper for inference focused on pairs of individuals or sampled units; for example, the substance use of one spouse in a shared household, conditioned on the substance use of the other spouse.  We formulate the pseudo likelihood with weights based on pairwise or second order probabilities and demonstrate consistency, removing the requirement for asymptotic independence and replacing it with restrictions on higher order selection probabilities.  Our approach provides a nearly automated estimation procedure applicable to any model specified by the data analyst.  We demonstrate our method on the National Survey on Drug Use and Health.
\end{abstract}


\end{frontmatter}

\section{Introduction} \label{motivation}
The primary interest of the data analyst is to perform inference about a finite population generated from an unknown model, $P_{0}$.   The observed data are collected from a sample taken from that finite population under a known sampling design distribution, $P_{\nu}$, that induces a correlation between the response variable of interest and the inclusion probabilities.  Sampling designs that induce this correlation are termed, ``informative", and the balance of information in the sample is different from that in the population.  \citet{2015arXiv150707050S} proposed an automated approach that formulates a sampling-weighted pseudo posterior density by exponentiating each likelihood contribution by a sampling weight constructed to be inversely proportional to its marginal inclusion probability, $\pi_{i} = P\left(\delta_{i} = 1\right)$, for units, $i = 1,\ldots,n$, where $n$ denotes the number of units in the observed sample.  The inclusion of unit, $i$, from the population, $U$, in the sample is indexed by $\delta_{i} \in \{0,1\}$.   Although we typically expect dependence to be induced among the sampled observations by $P_{\nu}$ - for example, under sampling without replacement - the use of weights composed from first order inclusion probabilities ignores this dependence; hence,
condition $(A5)$ in \citet{2015arXiv150707050S} restricts the class of sampling designs to those where the pairwise dependencies among units attenuate to $0$ in the limit of the population size, $N$, (at order $N$) to guarantee posterior consistency of the pseudo posterior distribution estimated on the sample data, at $P_{0}$ (in $L_{1}$).

While many sampling designs will meet this criterion, many won't; for example, a two-stage clustered sampling design where the number of clusters increases with $N$, but the number of units in each cluster remain relatively fixed such that the dependence induced at the second stage of sampling never attenuates to $0$.  A common example are designs which select households as clusters.

Researchers and policy makers may be interested in the relationship between the behaviors of individuals living together (such as parents and children or spouses). This creates a sub-population of individuals defined by the behaviors of \emph{other} members of the household, where these joint or conditional behaviors (such as substance use) are only observed through the survey.
Substance use, however, is not observed in the population, but only for respondents in the sample.  So the sub-population of interest is constructed by a conditioning event based on the reported substance use of other units in the sample.   Sampling weights defined based on marginal inclusion probabilities are formed using quantities (e.g., size variable(s)) observed for \emph{all units in the population} (within each stage of sampling)
and aren't designed to perform inference on a sub-population defined by information only available from \emph{other units in the sample}.

\subsection{Examples}\label{examples}
We next outline some examples of survey programs that employ informative sampling designs under which estimation using sampling weights formed from marginal inclusion probabilities would not be guaranteed to produce a consistent result under \citet{2015arXiv150707050S}.
\textit{Example 1:} The Current Expenditure (CE) survey is administered to U.S. households by the U.S. Bureau of Labor Statistics (BLS) for the purpose of determining the amount of spending for a broad collection of goods and service categories and it serves as the main source used to construct the basket of goods later used to formulate the Consumer Price Index.  The CE employs a multi-stage sampling design that draws clusters of core-based statistical areas (CBSAs), such as metropolitan and micropolitan areas, from which Census blocks and, ultimately, households are sampled.  Economists desire to model the propensity or probability of purchase for a variety of goods and services.
The CE sampling design is one where the number of clusters drawn increases in the limit of the population size, $N$, but the number of Census blocks per cluster remains relatively fixed such that we do \emph{not} expect an attenuation of the pairwise dependencies (induced within cluster) of the secondary sampling units (Census blocks) such that the pseudo posterior formulated from marginal inclusion probabilities would not be guaranteed to achieve a consistent result under this sampling design.

\textit{Example 2:}  The motivating survey for the pairwise weighting method that we introduce in this paper is the National Survey on Drug Use and Health (NSDUH), sponsored by the Substance Abuse and Mental Health Services Administration (SAMHSA). NSDUH is the primary source for statistical information on illicit drug use, alcohol use, substance use disorders (SUDs), mental health issues, and their co-occurrence for the civilian, non institutionalized population of the United States.  The NSDUH employs a multi-stage state-based design, with the earlier stages defined by geography within each state in order to select households (and group quarters) nested within these geographically-defined PSUs.  Individuals or pairs of individuals are subsequently sampled from selected households.  Viewing each household as a (mini) population, it is clear that the number of individuals residing in a household (of size, $N_{h}$) remains fixed in the limit of $N$, such that there is always unattenuated sampling dependence among those individuals.

Researchers and policy makers may be interested in the substance use of one member of a household -  for example, a household that includes two spouses living together (which we term, a ``spouse-spouse" household) -  based on the substance use of another member of the household (e.g., their spouse), which is only observed in a subset of the sample and not in the entire sample or the population.  Weights constructed on marginal inclusion probabilities may not map back to the sub-population (formed by conditioning on the self-reported behavior of the spouse) under informative sampling of the sub-population because the event (substance use by a spouse in the household) is only observed in the sample, whereas these weights are constructed only from quantities observed in the population.  We illustrate this potential problem of using marginal weights for sub-population inference based on self-reported alcohol use of household members from the NSDUH.

\subsection{Population Model Estimation}
The target audience for this article are data analysts who wish to perform some distributional inference
using data obtained from an informative sample design on a population using a model they specify, $p\left(y_{i}\vert \bm{\lambda}\right),~\bm{\lambda} \in \Lambda$, equipped with density, $p$.  We discuss, in the next section, how the limited literature on this topic does not adequately provide a
general method for making distributional inference on a population model formulated by the data analyst while
adjusting for the unequal probabilities of selection.

In this article, we propose an approach that replaces the pseudo likelihood of \citet{2015arXiv150707050S}, $p\left(y_{i}\vert \delta_{i} = 1,\bm{\lambda}\right)^{w_{i}}$, where the sampling weight, $w_{i} \propto 1/\pi_{i}$, with an approach that incorporates pairwise (second order) inclusion probabilities that provide some information about the dependence among sampled units induced by $P_{\nu}$.  The revised pseudo likelihood we will use in this paper is formed from pairwise terms, $\left[p\left(y_{i}\vert \delta_{i} = 1,\bm{\lambda}\right) \times p\left(y_{j}\vert \delta_{j} = 1,\bm{\lambda}\right)\right]^{w_{ij}}$, for $i,j \in U$, with $w_{ij} \propto 1/\pi_{ij},~\pi_{ij} = P\left(\delta_{i} = 1 \cap \delta_{j} = 1\right)$.  The use of weights constructed from pairwise inclusion probabilities conveys more information about the dependence induced by the joint sampling design distribution among the sampled units.  Our approach retains the attractive feature of \citet{2015arXiv150707050S} of asymptotically unbiased inference for $P_{0}$ under any model specified by the data analyst without altering the geometry of the Markov chain Monte Carlo (MCMC) sampler.  Our new approach also does not require the data analyst to have information about the sampling design, other than the (symmetric) matrix of pairwise inclusion probabilities for the lowest level units. Under many common designs, only a smaller block diagonal subset of this matrix may be needed.
The incorporation of second order inclusion probabilities, however, will broaden the class of sampling designs under which automated inference about $P_{0}$ may be performed by not restricting the pairwise dependence among the sampled units to attenuate to $0$.

\subsection{Review of  Methods to Account for Informative Sampling}
Research activity that incorporates sampling weights built from marginal (or first order) unit inclusion probabilities to estimate population quantities under an informative sample has surged.  Recent works by \citet{dong:2014, wu:2010, kunihama:2014, si2015} incorporate first order sampling weights, but under a single or fixed formulation for the population generation model - typically an empirical likelihood or Dirichlet process mixture for flexibility - with a focus on performing inference about simple population statistics, such as the total and mean.  These approaches focus on design inference, rather than inference from a model of interest specified by the data analyst, the latter of which is our focus in this paper.  \citet{2015arXiv150707050S}, alternatively, formulate a pseudo posterior distribution as a plug-in estimator, using first order sampling weights to allow the data analyst to perform inference from any population generating model that they specify.

As earlier mentioned, \citet{2015arXiv150707050S} define conditions that restrict allowable sampling designs such that frequentist consistency of their pseudo posterior approximation is guaranteed.  One of these conditions requires sampling designs where the pairwise sample inclusion dependencies among units attenuates to $0$ in the limit of the population size.  While they discuss many sampling designs that satisfy this restriction, many do not - such as the CE and NSDUH examples we earlier discussed.
In a similar fashion to this paper, \citet{2016yi} start with a pairwise likelihood construction that incorporates sampling weights based on second order (pairwise) inclusion probabilities to perform (pseudo) maximum likelihood estimation in order to capture second order dependence among sampled units.  However, they construct the population generating model to explicitly match the sampling design, which they restrict to a $2-$ stage sampling design.  So like the recent works using sampling weights composed from first order inclusion probabilities, they require a specific formulation of the population model that does not allow the data analyst to perform inference on a model of their choosing.  By contrast, our approach for incorporating sampling weights based on pairwise inclusion probabilities allows model inference under \emph{a large class of} population generating models specified by the data analyst. We do not require the use of a $2-$ stage sampling design or even that the population model and sampling designs match; rather, in the sequel we will formulate conditions that, together, define a class of sampling designs under which frequentist consistency of our (improper) pseudo posterior approximation is guaranteed.  Our use of second order (pairwise) weights allows us to broaden the class of allowable samplings designs under which our pseudo posterior estimator contracts on the true population generating distribution by eliminating the requirement for pairwise dependencies to attenuate to $0$.

\section{Pairwise Weighting to Account for Informative Sampling}\label{methods}
We begin by constructing the pseudo likelihood and associated pseudo posterior density under any analyst-specified prior formulation on the model,
$\bm{\lambda} \in \Lambda$.

Suppose there exists a Lebesgue measurable population-generating density,
\newline $\pi\left(y\vert\bm{\lambda}\right)$, indexed by parameters,
$\bm{\lambda} \in \Lambda$. Let $\delta_{i} \in \{0,1\}$ denote the sample inclusion indicator for units $i = 1,\ldots,N$ from the population under sampling without replacement.  The density for the observed sample is denoted by, $\pi\left(y_{o}\vert\bm{\lambda}\right) = \pi\left(y\vert \delta_{i} = 1,\bm{\lambda}\right)$, where ``$o$" indicates ``observed".

The following plug-in estimator for the posterior density incorporates sampling weights formulated from pairwise inclusion probabilities under the analyst-specified model for $\bm{\lambda} \in \Lambda$,
\begin{align}
p^{\pi}\left(\bm{\lambda}\vert \mathbf{y}_{o},\mathbf{w}\right) &\propto \left[\mathop{\prod}_{i, j = 1}^{n}\{p\left(y_{o,i}\vert \bm{\lambda}\right)p\left(y_{o,j}\vert \bm{\lambda}\right)\}^{w_{ij}}\right]\pi\left(\bm{\lambda}\right)\\
&\displaystyle\propto  \left[\mathop{\prod}_{i = 1}^{n}\mathop{\prod}_{j \neq i \in S}p\left(y_{o,i}\vert \bm{\lambda}\right)^{w_{ij}}\right]\pi\left(\bm{\lambda}\right)\\
&=  \left[\mathop{\prod}_{i = 1}^{n}p\left(y_{o,i}\vert \bm{\lambda}\right)^{\displaystyle\mathop{\sum}_{j\neq i \in S}w_{ij}}\right]\pi\left(\bm{\lambda}\right)\\\label{rearrange}
&=  \left[\mathop{\prod}_{i = 1}^{n}p\left(y_{o,i}\vert \bm{\lambda}\right)^{\displaystyle w^{\ast}_{i}}\right]\pi\left(\bm{\lambda}\right)
\end{align}
where we have used the independence of the $(y_{i})$, conditioned on $\bm{\lambda}$, under $P_{0}$, to rearrange terms in the product to achieve Equation~\ref{rearrange}, which exponentiates the likelihood contribution of unit $i$ by the sum of the sampling weights, $\{w_{ij} \propto 1/\pi_{ij}\}$, formulated to be inversely proportional to \emph{pairwise} or second order inclusion probabilities.  The collection of pairwise inclusion probabilities that, together, are used to formulate, $w^{\ast}_{i}$, the sampling weighted exponent for unit $i$, represent all pairs by which unit $i$ enters the observed sample.  The sum of the pairwise sampling weights for each unit, $i$, assigns the relative importance of the likelihood contribution for that observation to approximate the likelihood for the population.  We use $p^{\pi}$ to denote the noisy approximation to distribution, $p$, and we make note that the approximation is based on the data, $\mathbf{y}_{o}$ , and sampling weights, $\{\mathbf{w}^{\ast}\}$, confined to those units \emph{included} in the realized sample, $\{i \in U: \delta_{i} = 1\}$, where $U$ denotes a population of units indexed  by $i = 1,\ldots,N$.

The total estimated posterior variance is regulated by the sum of the sampling weights.  We define unnormalized second order weights, $\{w_{ij} = 1/\pi_{ij}\}$, and subsequently normalize them in two steps: 1. Define intermediate summed weight for unit $i$, $\tilde{w}_{i} = \mathop{\sum}_{j\neq i} w_{ij} / (n-1)$,  which we divide by $n-1$ to account for the $n-1$ times that the each $p\left(y_{o,i}\vert \bm{\lambda}\right)$ appears in Equation~\ref{rearrange}; 2. Construct $w^{\ast}_{i} = \frac{n}{\sum_{i=1}^{n} \tilde{w}_{i}} \tilde{w}_{i},~i = 1,\ldots,n$  to sum to the sample size, $n$.

The weight for an individual is formulated by a sum of the components constructed from the multiplication of (inverse) joint probabilities across the sampling stages. We will demonstrate in Section~\ref{simulation} that under a multi-stage sampling design, such as the sampling of individuals within households within geographic segments for our NSDUH example, this sum is dominated by terms that essentially factor due to nearly independent sampling in the earlier stages.  The result is that the pairwise-formulated weight, $w^{\ast}_{i}$, quickly converges (under increasing sample size) to the sampling weight formed from marginal inclusion probabilities, $w_{i}$.  So we will propose and discuss a modification that will reformulate $w^{\ast}_{i}$ to include only a single pairwise term, $w_{ij\vert\ell}$, in the case a single pair is sampled within a household, $\ell$, and individuals $j$ and and $i$ are co-sampled. Our reformulated pairwise weighting scheme treats households as a population of interest. It will be shown to reduce bias relative to the use of weights formulated from marginal (or individual) inclusion probabilities for inference about the behavior of one member of a pair conditioned on the behavior of the other under an informative selection of the conditioning event.

\subsection{Pseudo Posterior Distribution}\label{pseudop}
A sampling design is defined by placing a \emph{known} distribution on a vector of inclusion indicators, $\bm{\delta}_{\nu} = \left(\delta_{\nu 1},\ldots,\delta_{\nu N_{\nu}}\right)$, linked to the units comprising the population, $U_{\nu}$. Choice of design (i.e. specification of the distribution for $\bm{\delta}_{\nu}$) may depend on values from the population which is generated from a hypothetical distribution, $P_{0}$. The sampling distribution is subsequently used to take an \emph{observed} random sample of size $n_{\nu} \leq N_{\nu}$.
Our conditions needed for the main result, to follow, employ known second-order or pairwise unit inclusion probabilities, $\pi_{\nu ij} = \mbox{Pr}\{\delta_{\nu i} = 1 \cap \delta_{\nu j} = 1\}$ for all $i \neq j \in U_{\nu}$, rather than the marginal inclusion probabilities, $\pi_{\nu i} = \mbox{Pr}\{\delta_{\nu i}=1\}$ for $i\in U_{\nu}$ used in \citet{2015arXiv150707050S}, which are both obtained from the joint distribution over $\left(\delta_{\nu 1},\ldots,\delta_{\nu N_{\nu}}\right)$.  The dependence among unit inclusions in the sample contrasts with the usual $iid$ draws from $P$.  We denote the sampling distribution by $P_{\nu}$.

Under informative sampling, the inclusion probabilities (typically marginal)
are formulated to depend on the finite population data values, $\mathbf{X}_{N_{\nu}} = \left(\mbf{X}_{1},\ldots,\mbf{X}_{N_{\nu}}\right)$.
Since the resulting balance of information would be different in the sample, a posterior distribution for $\left(\mbf{X}_{1}\delta_{\nu 1},\ldots,\mbf{X}_{N_{\nu}}\delta_{\nu N_{\nu}}\right)$ that ignores the distribution for $\bm{\delta}_{\nu}$ will not lead to consistent estimation. In addition, under a complex sampling design with multiple stages, correlations are typically induced among the inclusions for some or all units.

Our task is to perform inference about the population generating distribution, $P_{0}$, using the observed data taken under an informative sampling design.  We account for informative sampling by ``undoing" the sampling design with the weighted estimator,
\begin{equation}\label{pseudopost}
p^{\pi}\left(\mbf{X}_{i}\delta_{\nu i}\right) := p\left(\mbf{X}_{i}\right)^{\frac{1}{(N_{\nu}-1)}\mathop{\sum}_{k\neq i \in U_{\nu}}\frac{\delta_{\nu i}\delta_{\nu k}}{\pi_{\nu ik}}},
\end{equation}
that weights each density contribution, $p(\mbf{X}_{i})$, by the sum of all of its inverse pairwise inclusion probabilities, which together represent all pairwise paths by which unit $i$ may enter a selected sample. The employment of pairwise inclusion probabilities partially accounts for the dependence of among unit inclusions induced by $P_{\nu}$.  The sum of terms for each $i$ is divided by $N_{\nu} - 1$ because each individual is present in $N_{\nu} -1$ population pair terms in the summation, each of which has expectation with respect to $P_{\nu}$ equal to $1$. So the normalization of the summation term ensures that the expectation of the logarithm of the density with respect to $P_{\nu}$ is unbiased. Our construction re-weights the likelihood contributions defined on those units randomly-selected for inclusion in the observed sample ($\{i \in U_{\nu}:\delta_{\nu i} = 1\}$) to \emph{approximate} the balance of information in $U_{\nu}$, from which we construct the associated pseudo posterior,
\begin{equation}\label{inform_post}
\Pi^{\pi}\left(B\vert \mbf{X}_{1}\delta_{\nu 1},\ldots,\mbf{X}_{N_{\nu}}\delta_{\nu N_{\nu}}\right) = \frac{\mathop{\int}_{P \in B}\mathop{\prod}_{i=1}^{N_{\nu}}\frac{p^{\pi}}{p_{0}^{\pi}}(\mbf{X}_{i}\delta_{\nu i})d\Pi(P)}{\mathop{\int}_{P \in \mathcal{P}}\mathop{\prod}_{i=1}^{N_{\nu}}\frac{p^{\pi}}{p_{0}^{\pi}}(\mbf{X}_{i}\delta_{\nu i})d\Pi(P)},
\end{equation}
that we use to achieve our required conditions for the rate of contraction of the pseudo posterior distribution on $P_{0}$.  We note that both $P$ and $\bm{\delta}_{\nu}$ are random variables defined on the space of measures ($\mathcal{P}$ and $ B \subseteq \mathcal{P}$) and possible samples, respectively.  An important condition on $P_{\nu}$ formulated in \citet{2015arXiv150707050S} that guarantees contraction of the pseudo posterior on $P_{0}$ restricts pairwise inclusion dependencies to asymptotically attenuate to $0$.  This restriction narrows the class of sampling designs for which consistency of a pseudo posterior based on marginal inclusion probabilities may be achieved.  We show in the sequel that our use of pairwise inclusion probabilities to formulate sampling weights in the pseudo posterior distribution replaces their condition that requires marginal factorization of the pairwise inclusion probabilities with two conditions that require pairwise factorization of \emph{third} and \emph{fourth} order inclusion probabilities. This expands the allowable class of sampling designs under which frequentist consistency may be guaranteed.  We assume measurability for the sets on which we compute prior, posterior and pseudo posterior probabilities on the joint product space, $\mathcal{X}\times\mathcal{P}$.  For brevity, we use the superscript, $\pi$, to denote the dependence on the known sampling probabilities, $\{\pi_{\nu ij}\}_{i,j \in U_{\nu}}$; for example,
\begin{multline*}
\displaystyle\Pi^{\pi}\left(B\middle\vert \mbf{X}_{1}\delta_{\nu 1},\ldots,\mbf{X}_{N_{\nu}}\delta_{\nu N_{\nu}}\right) := \Pi\left(B\middle\vert \left(\mbf{X}_{1}\delta_{\nu 1},\ldots,\mbf{X}_{N_{\nu}}\delta_{\nu N_{\nu}}\right)\vphantom{\mathop{\sum}_{k\neq 1 \in U_{\nu}}}\right.,\\{}\left.\left(\mathop{\sum}_{k\neq 1 \in U_{\nu}}\pi_{\nu 1k},\ldots,\mathop{\sum}_{k\neq N_{\nu} \in U_{\nu}}\pi_{\nu N_{\nu}k}\right)\vphantom{B}\right).
\end{multline*}

Our main result is achieved in the limit as $\nu\uparrow\infty$, under the countable set of successively larger-sized populations, $\{U_{\nu}\}_{\nu \in \mathbb{Z}^{+}}$.  We define the associated rate of convergence notation, $\order{b_{\nu}}$, to denote $\mathop{\lim}_{\nu\uparrow\infty}\frac{\order{b_{\nu}}}{b_{\nu}} = 0$.

\subsection{Empirical process functionals}\label{empirical}
We employ the empirical distribution approximation for the joint distribution over population generation and the draw of an informative sample that produces our observed data to formulate our results.  Our empirical distribution construction follows \citet{breslow:2007} and incorporates inverse inclusion pairwise probability weights, $\{1/\pi_{\nu ij}\}_{i, j \in U_{\nu}}$, to account for the informative sampling design,
\begin{equation}
\mathbb{P}^{\pi}_{N_{\nu}} = \frac{1}{N_{v}}\mathop{\sum}_{i=1}^{N_{\nu}}\frac{1}{(N_{\nu}-1)}\mathop{\sum}_{k\neq i \in U_{\nu}}\frac{\delta_{\nu i}\delta_{\nu k}}{\pi_{\nu ik}}\delta\left(\mbf{X}_{i}\right),
\end{equation}
where $\delta\left(\mbf{X}_{i}\right)$ denotes the Dirac delta function, with probability mass $1$ on $\mbf{X}_{i}$ and we recall that $N_{\nu} = \vert U_{\nu} \vert$ denotes the size of of the finite population. This construction contrasts with the usual empirical distribution, $\mathbb{P}_{N_{\nu}} = \frac{1}{N_{v}}\mathop{\sum}_{i=1}^{N_{\nu}}\delta\left(\mbf{X}_{i}\right)$, used to approximate $P \in \mathcal{P}$, the distribution hypothesized to generate the finite population, $U_{\nu}$.

We follow the notational convention of \citet{Ghosal00convergencerates} and define the associated expectation functionals with respect to these empirical distributions by $\mathbb{P}^{\pi}_{N_{\nu}}f = \frac{1}{N_{\nu}}\mathop{\sum}_{i=1}^{N_{\nu}}\frac{1}{(N_{\nu}-1)}\mathop{\sum}_{k\neq i \in U_{\nu}}\frac{\delta_{\nu i}\delta_{\nu k}}{\pi_{\nu ik}}f\left(\mbf{X}_{i}\right)$.  Similarly, $\mathbb{P}_{N_{\nu}}f = \frac{1}{N_{\nu}}\mathop{\sum}_{i=1}^{N_{\nu}}f\left(\mbf{X}_{i}\right)$.  Lastly, we use the associated centered empirical processes, $\mathbb{G}^{\pi}_{N_{\nu}} = \sqrt{N_{\nu}}\left(\mathbb{P}^{\pi}_{N_{\nu}}-P_{0}\right)$ and $\mathbb{G}_{N_{\nu}} = \sqrt{N_{\nu}}\left(\mathbb{P}_{N_{\nu}}-P_{0}\right)$.

The sampling-weighted, (average) pseudo Hellinger distance between distributions, $P_{1}, P_{2} \in \mathcal{P}$,
\begin{equation}
d^{\pi,2}_{N_{\nu}}\left(p_{1},p_{2}\right) = \frac{1}{N_{\nu}}\mathop{\sum}_{i=1}^{N_{\nu}}\frac{1}{(N_{\nu}-1)}\mathop{\sum}_{k\neq i \in U_{\nu}}\frac{\delta_{\nu i}\delta_{\nu k}}{\pi_{\nu ik}}d^{2}\left(p_{1}(\mathbf{X}_{i}),p_{2}(\mathbf{X}_{i})\right),
\end{equation}
where $d\left(p_{1},p_{2}\right) = \left[\mathop{\int}\left(\sqrt{p_{1}}-\sqrt{p_{2}}\right)^{2}d\mu\right]^{\frac{1}{2}}$ (for dominating measure, $\mu$).
We need this empirical average distance metric because the observed (sample) data drawn from the finite population under $P_{\nu}$ are no longer independent.  The implication is that our consistency result applies to finite populations generated as $inid$ from which informative samples are taken.  The associated non-sampling Hellinger distance is specified with, $d^{2}_{N_{\nu}}\left(p_{1},p_{2}\right) = \frac{1}{N_{\nu}}\mathop{\sum}_{i=1}^{N_{\nu}}d^{2}\left(p_{1}(\mathbf{X}_{i}),p_{2}(\mathbf{X}_{i})\right)$.

\subsection{Main result}\label{results}
We proceed to construct associated conditions and a theorem that contain our main result on the consistency of the pairwise pseudo posterior distribution under a class of informative sampling designs at the true generating distribution, $P_{0}$.  Our approach extends the main in-probability convergence result of \citet{ghosal2007} by adding new conditions that restrict the distribution of the informative sampling design.  Suppose we have a  sequence, $\xi_{N_{\nu}} \downarrow 0$ and $N_{\nu}\xi^{2}_{N_{\nu}}\uparrow\infty$  and $n_{\nu}\xi^{2}_{N_{\nu}}\uparrow\infty$ as $\nu\in\mathbb{Z}^{+}~\uparrow\infty$ and any constant, $C >0$,

\begin{description}
\item[(A1)\label{existtests}] (Local entropy condition - Size of model)
        \begin{equation*}
        \mathop{\sup}_{\xi > \xi_{N_{\nu}}}\log N\left(\xi/36,\{P\in\mathcal{P}_{N_{\nu}}: d_{N_{\nu}}\left(P,P_{0}\right) < \xi\},d_{N_{\nu}}\right) \leq N_{\nu} \xi_{N_{\nu}}^{2},
        \end{equation*}
\item[(A2)\label{sizespace}] (Size of space)
        \begin{equation*}
        \displaystyle\Pi\left(\mathcal{P}\backslash\mathcal{P}_{N_{\nu}}\right) \leq \exp\left(-N_{\nu}\xi^{2}_{N_{\nu}}\left(2(1+2C)\right)\right)
        \end{equation*}
\item[(A3)\label{priortruth}] (Prior mass covering the truth)
        \begin{equation*}
        \displaystyle\Pi\left(P: -P_{0}\log\frac{p}{p_{0}}\leq \xi^{2}_{N_{\nu}}\cap P_{0}\left[\log\frac{p}{p_{0}}\right]^{2}\leq \xi^{2}_{N_{\nu}} \right) \geq \exp\left(-N_{\nu}\xi^{2}_{N_{\nu}}C\right)
        \end{equation*}
\item[(A4)\label{bounded}] (Non-zero Pairwise Inclusion Probabilities)
        \begin{equation*}
        \displaystyle\mathop{\sup}_{\nu}\left[\frac{1}{\displaystyle\mathop{\min}_{i,k:k\neq i\in U_{\nu}}\vert\pi_{\nu ik}\vert}\right] \leq \gamma \geq 1, \text{  with $P_{0}-$probability $1$.}
        \end{equation*}
\item[(A5)\label{factorthird}] (Bounded Ratio of Third to Second Order Inclusion Probabilities)
        \begin{align*}
        &\displaystyle\mathop{\sup}_{\nu}\mathop{\max}_{i,k,\ell: k\neq\ell\neq i\in U_{\nu}}\biggl\vert\frac{\pi_{\nu ik\ell}}{\pi_{\nu ik}\pi_{\nu i\ell}}\biggl\vert \\
        &= \displaystyle\mathop{\sup}_{\nu}\mathop{\max}_{i,k,\ell: k\neq\ell\neq i\in U_{\nu}}\biggl\vert\frac{\pi_{\nu k\ell\vert i}}{\pi_{\nu k\vert i}\pi_{\nu \ell\vert i}\pi_{i}}\biggl\vert \leq C_{5}, \text{  with $P_{0}-$probability $1$,}
        \end{align*}
        \vskip -0.2in
        where
        \begin{equation*}
        \pi_{\nu k\vert i} = \mbox{Pr}\left(\delta_{\nu k}=1\vert \delta_{\nu i} = 1\right),~
        \pi_{\nu k\ell\vert i} = \mbox{Pr}\left(\delta_{\nu k}=1\cap\delta_{\nu\ell}=1\vert \delta_{\nu i} = 1\right).
        \end{equation*}
\item[(A6)\label{factorfourth}] (Asymptotic Factorization of Fourth Order Inclusion Probabilities)
        \begin{equation*}
        \displaystyle\mathop{\limsup}_{\nu\uparrow\infty} \mathop{\max}_{i,j,k,\ell: i \neq j, k\neq i, \ell\neq j\in U_{\nu}}\left\vert\frac{\pi_{\nu ikj\ell}}{\pi_{\nu ik}\pi_{\nu j\ell}} - 1\right\vert = \order{N_{\nu}^{-1}}, \text{  with $P_{0}-$probability $1$}
        \end{equation*}
        such that for some constant, $C_{4} > 0$,
        \begin{equation*}
        \displaystyle N_{\nu}\mathop{\sup}_{\nu}\mathop{\max}_{i,j,k,\ell: i \neq j, k\neq i, \ell\neq j\in U_{\nu}}\left\vert\frac{\pi_{\nu ikj\ell}}{\pi_{\nu ik}\pi_{\nu j\ell}} - 1\right\vert \leq C_{4}, \text{  for $N_{\nu}$ sufficiently large.}
        \end{equation*}
\item[(A7)\label{fraction}] (Constant Sampling fraction)
        For some constant, $f \in(0,1)$, that we term the ``sampling fraction",
        \begin{equation*}
        \mathop{\limsup}_{\nu}\displaystyle\biggl\vert\frac{n_{\nu}}{N_{\nu}} - f\biggl\vert = \order{1}, \text{  with $P_{0}-$probability $1$.}
        \end{equation*}
\end{description}
The first three conditions are the same as for \citet{2015arXiv150707050S} and restrict the growth rate of the model space (e.g., of parameters) and requires prior mass to be placed on an interval containing the true value.
The next four new conditions impose restrictions on the sampling design and associated known distribution, $P_{\nu}$, which are similar than those specified in \citet{2015arXiv150707050S}, but allow for a wider class of sampling designs under which consistency of the pseudo posterior formulation of Equation~\ref{pseudopost} is guaranteed by replacing the asymptotic attenuation of pairwise inclusion dependencies with restrictions on third and fourth order inclusion dependencies.  Condition~\nameref{bounded} requires the sampling design to assign a positive probability for pairwise inclusion for every pair of units, $i,j \in U_{\nu}$. Since the maximum pairwise inclusion probability is $1$, the bound, $\gamma \geq 1$. This condition is no more restrictive than the analogous condition $A4$ in \citet{2015arXiv150707050S}, which bounds marginal inclusion probabilities away from $0$, in the case that $\Cov\left(\delta_{\nu i},\delta_{\nu j}\right) > 0$, which implies that $\min\{\pi_{\nu i},\pi_{\nu j}\} \geq \pi_{\nu ij}  >  \pi_{\nu i}\pi_{\nu j}$; otherwise, for designs where $\Cov\left(\delta_{\nu i},\delta_{\nu j}\right) < 0$, condition~\nameref{bounded} \emph{is} more restrictive because $\{\pi_{\nu i},\pi_{\nu j}\}  > 0$ does \emph{not} imply $\vert\pi_{\nu ij}\vert > 0 $. All pairs of units must be assigned non-zero pairwise inclusion probabilities. We make note that other than this restriction bounding pairwise inclusion probabilities away from $0$, there is no required attenuation of pairwise dependencies as there is in \citet{2015arXiv150707050S}.  Instead, we add the new condition~\nameref{factorthird} that restricts sampling designs under which the ratio of third order inclusion probabilities to the product of second order inclusion probabilities is absolutely bounded from above.  This ratio approaches the condition of bounding first order inclusion probabilities away from $0$ in the case that the \emph{conditional} pairwise inclusion probabilities asymptotically factor (though such is not required).  Condition~\nameref{factorfourth} requires fourth order inclusion probabilities to factor to pairwise probabilities as $N_{\nu}\uparrow\infty$.  We note the presence of pairwise inclusion probabilities in the denominator for each our conditions \nameref{factorthird} and \nameref{factorfourth}, as contrasted with marginal inclusion probabilities in the analogous condition $A5$ in \citet{2015arXiv150707050S} (which requires asymptotic factorization of pairwise inclusion probabilities).  The conditions of \citet{2015arXiv150707050S} may be viewed as requiring sampling designs that limit to the equivalent to the independent sampling of individual units, while our conditions asymptotically require designs to limit to the independent sampling of \emph{pairs} of individuals. Condition~\nameref{fraction} ensures that the observed sample size, $n_{\nu}$, limits to $\infty$ along with the size of the partially-observed finite population, $N_{\nu}$, such that the variation of information about the population expressed in realized samples is controlled.

\begin{theorem}
\label{main}
Suppose conditions ~\nameref{existtests}-\nameref{fraction} hold.  Then for sets $\mathcal{P}_{N_{\nu}}\subset\mathcal{P}$, constants, $K >0$, and $M$ sufficiently large,
\begin{align}\label{limit}
&\mathbb{E}_{P_{0},P_{\nu}}\Pi^{\pi}\left(P:d^{\pi}_{N_{\nu}}\left(P,P_{0}\right) \geq M\xi_{N_{\nu}} \vert \mbf{X}_{1}\delta_{\nu 1},\ldots,\mbf{X}_{N_{\nu}}\delta_{\nu N_{\nu}}\right) \leq\nonumber\\
&\frac{16\gamma^{2}\left[\gamma+C_{3}\right]}{\left(Kf + 1 - 2\gamma\right)^{2}N_{\nu}\xi_{N_{\nu}}^{2}} + 5\gamma\exp\left(-\frac{K n_{\nu}\xi_{N_{\nu}}^{2}}{2\gamma}\right),
\end{align}
which tends to $0$ as $\left(n_{\nu}, N_{\nu}\right)\uparrow\infty$.
\end{theorem}
\begin{proof}
The proof follows exactly that in \citet{2015arXiv150707050S} where we bound the numerator (from above) and the denominator (from below) of the expectation with respect to the joint distribution of population generation and the taking of a sample of the pseudo posterior mass placed on the set of models, $P$, at some minimum pseudo Hellinger distance from $P_{0}$.  We replace their condition (A4), which bounds the inverse of marginal inclusion probabilities, with our condition ~\nameref{bounded}, that now bounds the inverse of pairwise inclusion probabilities.   We reformulate two enabling lemmas of \citet{2015arXiv150707050S}, which we present in an Appendix, where the reliance on (their) condition (A5) requiring asymptotic factoring of pairwise unit inclusion probabilities is here replaced by conditions ~\nameref{factorthird} and ~\nameref{factorfourth} that require asymptotic pairwise factoring of fourth order inclusion probabilities and boundedness in the ratio of third-to-second order inclusion probabilities.
\end{proof}

We note that the rate of convergence is decreased for a sampling distribution, $P_{\nu}$, that expresses a large variance in unit pairwise inclusion probabilities such that $\gamma$ will be relatively larger. Samples drawn under a design that expresses a large variability in the second order sampling weights will express more dispersion in their information relative to a simple random sample of the underlying finite population.  We construct $C_{3} = C_{4} + C_{5} + 1$, such that to the extent that the third and fourth order dependencies attenuate faster than the pairwise inclusion probabilities under the pseudo posterior constructed from first order sampling weights, then the rate of contraction will be faster under our formation than in \citet{2015arXiv150707050S}.  In general, however, one would not necessarily expect a more rapid contraction under our employment of second order inclusion probabilities to form our sampling weights because the rate in both \citet{2015arXiv150707050S} and here is nearly optimal, as we may observe by plugging in for the rate, $\xi_{N_{\nu}} = \log n_{\nu}/\sqrt{n_{\nu}}$ - the optimal convergence rate reduced by a log factor -  and noting that the bound in Equation~\ref{limit} limits to $0$.  The main benefit of our approach is that it is expected to broaden the class of sampling designs (relative to \citet{2015arXiv150707050S}) under which the associated pseudo posterior distribution achieves a frequentist consistency result.

\section{Population Model}\label{sec:model}
We construct a population model to address our inferential interest of assessing the functional form of the relationship between frequency of alcohol consumption and age at conditional quantiles of interest for the population distribution of the U.S., as estimated from the 2014 National Survey on Drug Use and Health (NSDUH).

We follow \citet{Reed09apartially} and formulate a likelihood for each observation using the asymmetric Laplace (AL) distribution,
\begin{equation} \label{pop_like}
y_{i}\mid \mu_{i},\tau, q \ind \mathcal{AL}\left(\mu_{i}, \tau, q\right),~ i = 1,\ldots,N
\end{equation}
where $\tau$ is a precision parameter and $q\in (0,1)$ is the quantile of interest.  We recall the AL density for observed response, $y$,
\begin{equation}
p\left(y\mid \mu, \tau, q\right) = \tau  q(1-q) \exp\left(-\tau \rho_{q}(y - \mu)\right),
\end{equation}
where
\begin{equation}
   \rho_{q}\left(u\right) :=
    \begin{cases}
      q|u|, & \text{if } u \geq 0 \\
      \left(1-q\right)|u|, & \text{if } u <0
    \end{cases}
\end{equation}

To accommodate expected non-linearity in the relationship of age with the distribution for alcohol consumption, we specify a B-spline basis term,
\begin{equation}
\bm{\mu} = \mathbf{B}\bm{\theta}
\end{equation}
for $N\times (d+k)$ B-spline basis matrix, $\mathbf{B}$, that we extend as in \citet{quteprints72987}, to convert the B-spline to a penalized (P-) spline of order $k$ with employment of a penalty matrix, $\mathbf{Q} = \mathbf{D}^{'}\mathbf{D}$, where $\mathbf{D}$ has $d+k$ columns for a B-spline basis with $d$ knots and is the discretized $k^{\mbox{\tiny{th}}}$ difference operator.  Higher values for $k$ enforce greater smoothness restrictions in each B-spline piecewise basis (column of $\mathbf{B}$) under the following multivariate Gaussian prior for $\bm{\theta}$,
\begin{equation}
p\left(\bm{\theta}\mid \lambda\right) \propto \exp\left(-\frac{\lambda}{2}\bm{\theta}^{'}\mathbf{Q}\bm{\theta}\right),
\end{equation}
where $(d+k)\times (d+k)$ penalty (precision) matrix, $\mathbf{Q}$ is of rank, $d$, in a similar fashion as the intrinsic conditional autoregressive prior \citep{rue:held:2005}.  Parameter, $\lambda$, is the smoothing, penalty parameter on which we impose a further $\mathcal{G}\left(1,1\right)$ prior, specified with small hyperparameter settings easily overwhelmed by the data.  We choose $d = 10$ and $k=3$ such that each spline basis lies in the space of piecewise $C^{3}$ functions. Precision parameter, $\tau$, from Equation~\ref{pop_like} also receives a $\mathcal{G}\left(1,1\right)$ prior.

We formulate the logarithm of the sampling-weighted pseudo likelihood for estimating $(\bm{\mu},\tau,\lambda)$ from our observed data for the $ n\leq N$ sampled units,
\begin{align}\label{pseudo_like}
\log\left[\mathop{\prod}_{i=1}^{n} p\left(y_{i}\mid \mu_{i},\tau,q\right)^{w^{\ast}_{i}}\right] &= \mathop{\sum}_{i=1}^{n}w^{\ast}_{i}\log p\left(y_{i}\mid \mu_{i},\tau,q\right)\\
&= w^{\ast}_{\mbox{\tiny{TOT}}}\left[\log \tau + \log q + \log (1-q)\right] \nonumber\\ &-\tau\mathop{\sum}_{i=1}^{n}w^{\ast}_{i}\rho_{q}\left(y_{i}-\mu_{i}\right),
\end{align}
where $\displaystyle w^{\ast}_{\mbox{\tiny{TOT}}} = \mathop{\sum}_{i=1}^{n}w^{\ast}_{i}$, with sampling weights, $w^{\ast}_{i}$, as defined using joint inclusion probabilities for unit $i$ in Section~\ref{pseudop} or, alternatively, using marginal inclusion probabilities as in \citet{2015arXiv150707050S}, to support our comparison of alternative weighting schema.  We recall that we have normalized the sum of the weights such that $w^{\ast}_{\mbox{\tiny{TOT}}} = n$.  Finally, we estimate the joint posterior distribution using Equation~\ref{pseudo_like}, coupled with our prior distributions assignments, using the NUTS Hamiltonian Monte Carlo algorithm implemented in Stan \citep{stan:2015}.

\section{Simulation Study}\label{simulation} 

\subsection{Scenarios}
We begin by abstracting the five-stage, geographically-indexed NSDUH sampling design \citep{MRB:Sampling:2014} to a simpler, three stage design (of \{area segment, household, individual\}) that we use to draw samples from a synthetic population in a manner that still generalizes to the NSDUH (and similar multi-stage sampling designs where the number of last stage units does not grow with overall population size).   We focus our inference on the case of analyzing (some conditional quantile for) alcohol usage for a sub-population that is formed by conditioning the inclusion of a sampled individual in a spouse-spouse household based on the self-reported frequency of alcohol usage by their spouse.   We construct three scenarios, where each targets a sub-population, under which we will compare the estimation performances (through bias and mean square error) of marginal versus pairwise weighting schema.  These scenarios will be used for both our simulation study and following application to NSDUH:
\begin{description}
\item[(S1)\label{subpop}] A sub-population target for inference is defined by those individuals who reside in a particular household configuration.  For ease-of-understanding and to tie back to inference on the NSDUH, let's suppose the household configuration of interest is spouse-spouse pairs. So we only include the sub-sample of individuals who reside in a spouse-spouse pair for model estimation (regardless of whether their spouse is also included in the sample).  This sub-population is formed using information observed in the (household) population; e.g., the household roster provides information on whether someone resides with a spouse.
\item[(S2)\label{obssubpop}] Our focus for inference continues to be the sub-population of individuals living in spouse-spouse pairs, but we more narrowly include a smaller sub-sample of individuals drawn from spouse-spouse pairs where (responses for) both are \emph{observed} in the sample.  The smaller sub-sample of spouse-spouse pairs that results from including only those spouses mutually observed will be designed to be informative; that is, the age distribution for individuals who are jointly observed with their spouses in the sub-sample will be different than the age distribution of individuals in the larger sample whose spouses are not co-included in the sample.  In practice, the data analyst would not use this contrived sub-sample since the larger sub-sample of scenario~\nameref{subpop} maps back to the same sub-population of interest.  We include this scenario both because it allows us to compare how the pairwise and marginal weighting methods adjust for informative sub-sampling and also because it sets the stage for further constricting the sub-population of interest by conditioning on an event only observed in the sample.
\item[(S3)\label{condsubpop}] The sub-population of interest is further restricted to those spouse-spouse households where one spouse consumes alcohol above (and/or below) some threshold level frequency.  This sub-population is defined based on a conditioning event (the level of alcohol consumption by one member of a spouse-spouse pair) and the condition is \emph{not} observed in the (household) population.  By construction, the associated sub-sample for each conditioning event of interest would be a subset of the sub-sample included in scenario~\nameref{obssubpop} because the conditioning event is restricted to be observed \emph{only} in the pair sample.
\end{description}

\subsection{Population Generation}
We simulate a population of N = 6000, with 200 primary sampling units (PSUs) each containing 10 households (HHs) which each contain 3 individuals (P1, P2, P3).
The response $y$ is drawn from an AL distribution with $q = 0.5$. We choose $\tau = 8$ to yield a relatively precise response. We let $\mu$ depend on two predictors $x_1$ and $x_2$. The variable $x_1$ represents the observed information available for analysis, whereas $x_2$ represents information available for sampling, which is either ignored or not available for analysis. The $x_1$ and $x_2$ distributions for P1 and P3 are $\mathcal{N}(0,1)$ and $\mathcal{E}(r =1/5)$ with rate $r$, where $\mathcal{N}(\cdot)$ and $\mathcal{E}(\cdot)$ represent normal and exponential distributions, respectively. The size measure used for sample selection is  $\tilde{x}_{2,ijk} = x_{2,ijk} - \min (x_{2,ijk}) + 1$ for $i = 1, \ldots, 3$ individuals, $j = 1, \ldots 10$ HHs, and $k = 1, \ldots, 200$ PSUs. The conditional quantile for P1 and P3 ($i = 1,3$) within each HH $j$ and PSU $k$ is
\[
\mu_{i} = 10 + 1 x_{1,i} + 0.5 x_{2,i} + 0.5 x_{1,i} x_{2,i} - {x_{1,i}^2}
\]

P2 is given a distribution for $x_{2,2jk}$ that depends on P1's value for $x_{2,1jk}$ within HH $j$.
Within each PSU, $k$, and HH, $j$, the distributions for $x_2$ ($i = 2$) are
$x_{2,2}|x_{2,1} \sim \mathcal{E}(1/x_{2,1})$, so $E(x_{2,2}|x_{2,1}) = x_{2,1}$.  The distribution for $\mu_2$ is further set to depend on whether the $x_2$ value for P1 in the same HH is higher or lower than the median $Q_{0.5}$ of $x_2$ among the population of P1s.
\[
\mu_{2} = \left\{
	\begin{array}{ll}
		10 + 1 x_{1,2} + 0.25 x_{2,2} + 0.25 x_{1,2} x_{2,2} - 2 x_{1,2}^2 ;& x_{2,1} \le Q_{0.5}(x_{2,1jk})\\
		10 + 1 x_{1,2} + 0.75 x_{2,2} + 0.75 x_{1,2} x_{2,2}; &  x_{2,1} > Q_{0.5}(x_{2,1jk})
	\end{array}
\right.
\]
In terms of P2-P1 and P2-P3 pairs within each HH, there are now different distributions for both the outcome $y_{2,jk}$ (via conditional $\mu_2$) and the joint selection probability (via conditional size $x_{2,2jk}|x_{2,1jk}$) even though the marginal distributions for outcomes $y_{1jk}$, $y_{3jk}$ and size measures $x_{2,1jk}$, $x_{2,3jk}$ are the same.  The conditioning of the value of the size variable for P2 on that for P1, in each household, together with constructing the form for the conditional quantile for P2, $\mu_{2}(x_{1,2},x_{2,2})$, based on thresholding the value of size variable, $x_{2,1}$, for P1, instantiates an informative sampling design of the sub-population of P2 individuals conditioned on the response values for P1 individuals.

Even though the population response $y$ was simulated with $\mu = f(x_1,x_2)$, we estimate the marginal models at the population level for $\mu = f(x_1)$ as described in section \ref{sec:model}. This exclusion of $x_2$ is analogous to the situation in which an analyst does not have access to all the sample design information and ensures that our sampling design instantiates informativeness (where $y$ is correlated with the selection variable, $x_{2}$, that defines inclusion probabilities). In particular, we estimate the models under each of three scenarios and compare the population fitted models, $\mu = f(x_1)$, to those from the samples.

\subsection{Sampling from the Population}
For the simulation, the number of selected PSUs was varied $K \in \{10, 20, 40, 80, 160\}$, the number of HHs within each PSU was fixed at 5, and the number of selected individuals within each HH was 2 (a pair). Each setting was repeated $M = 200$ times. Details for the selection at each stage follows:
\begin{enumerate}
	\item For each PSU indexed by $k$, an aggregate size measure $X_{2,k} = \sum_{ij} x_{2,ij|k}$ was created summing over all individuals $i$ and HHs $j$ in PSU $k$. PSUs are then selected proportional to this size measure based on Brewer's PPS algorithm \citep{BrewerPPS}.
	\item Once PSUs are selected, for each HH within the selected PSUs indexed by $j$ an aggregate size measure $X_{2,j|k} = \sum_{i} x_{2,i|jk}$ was created summing over all individuals $i$ within each HH in the selected PSUs. HHs are selected independently across PSUs. Within each PSU, HHs are selected proportional to size based on Brewer's PPS algorithm.
	\item Within each selected HH, a pair of persons (2 out of P1, P2, P3) is selected jointly. Firstly, all $ {3 \choose 2} = 6$ pairs are given a size equal to the sum of the individual size measures. So $X_{2,ii'|jk} = x_{2,i|jk}+ x_{2,i'|jk}$. Secondly, a single pair is then directly selected with probability proportional to this size measure. Individual (marginal) probabilities of selection for each of P1, P2, and P3 can be computed directly from the 6 pair inclusion probabilities.
\end{enumerate}

\subsection{Calculation of Weights}
From the three stages of sampling there are four weight components available to use:
\begin{enumerate}
\item PSU: $w_1^{k} = 1/\pi_k$, the inverse of the probability of selecting PSU $k$.
\item HH: $w_2^{j|k} = 1/\pi_{j|k}$, the inverse of the conditional probability of selecting HH $j$ given PSU $k$ has been selected.
\item Individual: $w_3^{i|jk} = 1/\pi_{i|jk}$, the inverse of the conditional probability of selecting individual $i$ given HH $j$ and PSU $k$ are selected.
\item Pairwise: $w_3^{i,i'|jk} = 1/\pi_{i,i'|jk}$, the inverse of the joint probability of selecting individuals $i$ and $i'$ as a pair in HH $j$ given the household and PSU $k$ are selected.
\end{enumerate}
In general, each stage could have two sets of weights from both first and second order components, but for this example the first two stages are sampled via PPS and thus their joint probabilities of selection within each stage are considered negligible.

Based on these four weight components, the first order or marginal weight is simply the inverse of the probability of selecting an individual: $w_{i}^{(1)} = w_1^{k} w_2^{j|k}w_3^{i|jk}$.  For second order weights, we set the HH as the unit of analysis and construct each pairwise weight \emph{within} HH for individual, $i$: $w^{(2p)}_{i} =w_1^{k} w_2^{j|k}w_3^{i,i'|jk} /(N_{p_j} -1)$, where $i^{'}$ is the co-sampled individual in HH, $j$, that includes units, $(i,i^{'})$.  We normalize by the number of pairs in the domain of interest within each household, $(N_{p_j})$, because each roster of the HH is treated as a population (and the entire population is constructed as the collection of household populations). Full pairwise (second order) weights, by contrast, are constructed by summing the inverse pairwise inclusion probabilities across \emph{all} individuals in the sample included with $i$ to focus on the entire population (across the collection of household populations) of size, $N$, as the unit of analysis.  Figure \ref{fig:weightviolin} compares the distribution of sampling weights under marginal, full pairwise and within-household pairwise weighting, from left-to-right, within each plot panel.  The panels in each column present distributions for realized samples of increasing size, from left-to-right.  The rows compare the weight distributions for all $\mbox{P}2$ units, P2 units where P1's response $<$ 10, and P2 units where P1's response $\ge$ 10, from top-to-bottom.

We observe that while the weight distributions are highly similar for marginal and full pairwise weighting, on the one hand, there are notable differences between the first two and household pairwise weighting, on the other hand.  Figure~\ref{fig:weightratio} plots the distributions for the \emph{ratio} of full and within household pairwise weight, to better understand the differences in their underlying distributions.  Taken together, both figures reveal that only the household pairwise weights actually redistribute the marginal weight, whereas the full sample (second order) pairwise weights quickly collapse to the marginal weights. The full pairwise weights converge to the marginal weights because the majority of terms in each summation to construct a weight value for each individual are from pairings across different PSUs and HHs. These terms are dominated by the early, nearly independent sampling stages and thus the small number (only one for pair samples) of within HH components provide negligible contributions to the sum. See Appendix \ref{sec:weightcalc} for more discussion on the formulation of the full and household pairwise weights.

\subsection{Results}
For simplicity and scalability to small sample sizes, we model both the population and sample using $d = 5$ knots and polynomials of degree $k = 2$. Each column of Figure \ref{fig:P2allsamp} displays the fitted curves, bias and mean square error (MSE) for scenario \nameref{subpop} that includes the full sample of P2's for a particular average sample size. Both weighting methods remove the bias compared to the equal weighting. The first order or marginal weights show a slight, but persistent edge in MSE likely due to less variability in the marginal weights (See Figures \ref{fig:weightviolin} and \ref{fig:weightratio}).

Figures \ref{fig:P2condPgt} and \ref{fig:P2condPlt} present results under scenario \nameref{condsubpop}, where the P1-P2 sub-population of interest is \emph{conditioned} on whether the observed response, $y$, for P1 is above or below $10$ (a value which is close to the median of $y$), respectively.  The household pairwise weights remove more bias and lead to smaller MSE than do the marginal weights because the conditioning event on the $y$ value for P1 is only observed in the P1-P2 pair sub-sample, but not in the full P2 sample or the (household) population.
Therefore, the computed value of each household pairwise weight is able to adjust for informative sub-sampling to more fully remove bias than marginal weighting. By contrast, the marginal weight for each individual is constructed based on quantities observed in the entire population, so it does not change or adapt to the particular sub-sample needed to study a conditioning event not observed in the population; that is, the marginal weight for each individual is fixed to the same value for every sample or pair sub-sample that includes this individual.  While the marginal weighting scheme does demonstrate a notable improvement in estimation bias and MSE compared to the unweighted case, much of this improvement may be due to the informativeness of the first two stages of PPS sampling. Using the marginal weights may lead to different, and potentially incorrect, inferential conclusions about the sub-population, as we will demonstrate in the Application section \ref{sec:NSDUH} that follows.

We also realize an improvement in bias and MSE performance under household pairwise weighting as compared to marginal weighting for scenario \nameref{obssubpop}, as shown in Figure~\ref{fig:P2margP1}. Since this P2 outcome model is \emph{not} formed under a conditioning event, one might expect marginal weighting to perform similarly to pair weighting for inference about the P2 sub-population; however, since the joint response distributions and joint selection probabilities of the P1-P2 and P2-P3 pairs differed, the P2 sub-sample selected in a P1-P2 pair differed from that selected in P2-P3 pairs. Therefore the marginal weights for P2 do not fully adjust for this additional sub-selection.

\begin{figure}
\centering
\includegraphics[width = 0.95\textwidth,
		page = 1,clip = true, trim = 0in 0in 0in 0in]{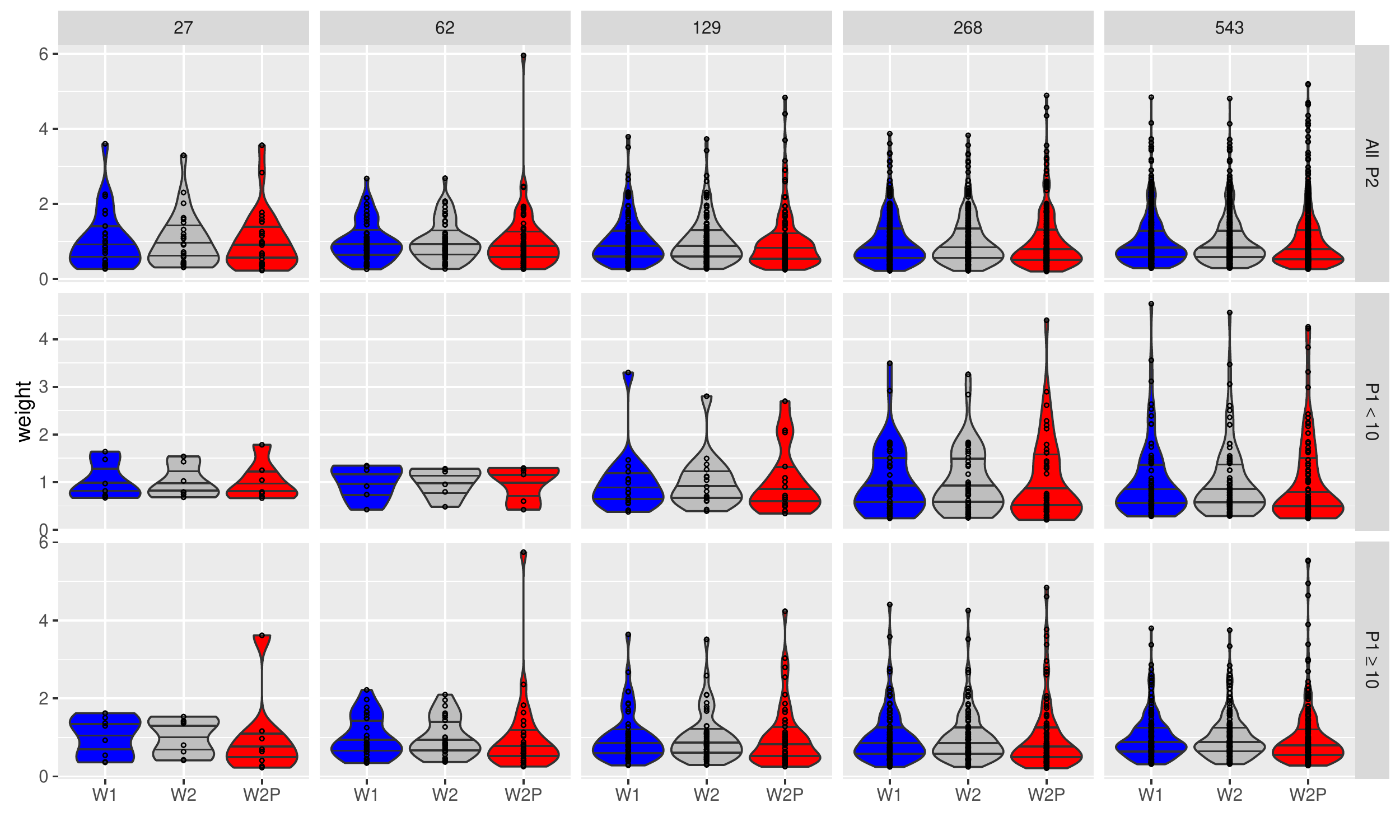}
\caption{Distributions of marginal (blue), full second order (grey), and within HH second order (red) weights for P2 across realizations of different sample sizes (column heading) and by subdomain (top to bottom) all P2, P2 where P1's response $<$ 10, P2 where P1's response $\ge$ 10}
\label{fig:weightviolin}
\end{figure}

\begin{figure}
\centering
\includegraphics[width = 0.95\textwidth,
		page = 2,clip = true, trim = 0in 0in 0in 0in]{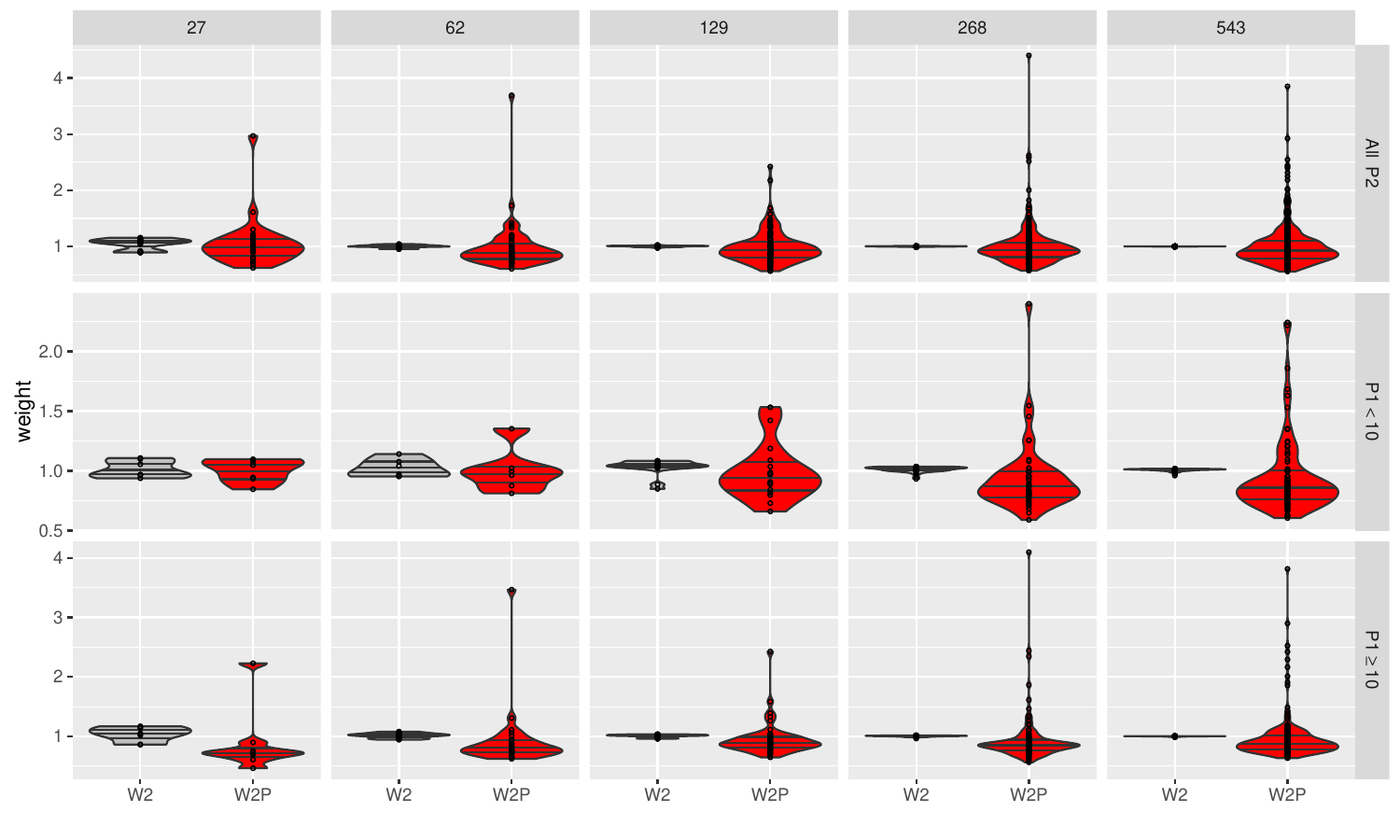}
\caption{Distributions of the ratio of full second order (grey) and within HH second order (red) weights to marginal weights for P2 across realizations of different sample sizes (column heading) and by subdomain (top to bottom) all P2, P2 where P1's response $<$ 10, P2 where P1's response $\ge$ 10}
\label{fig:weightratio}
\end{figure}

\begin{figure}
\centering
\includegraphics[width = 0.95\textwidth,
		page = 1,clip = true, trim = 0in 0in 0in 0in]{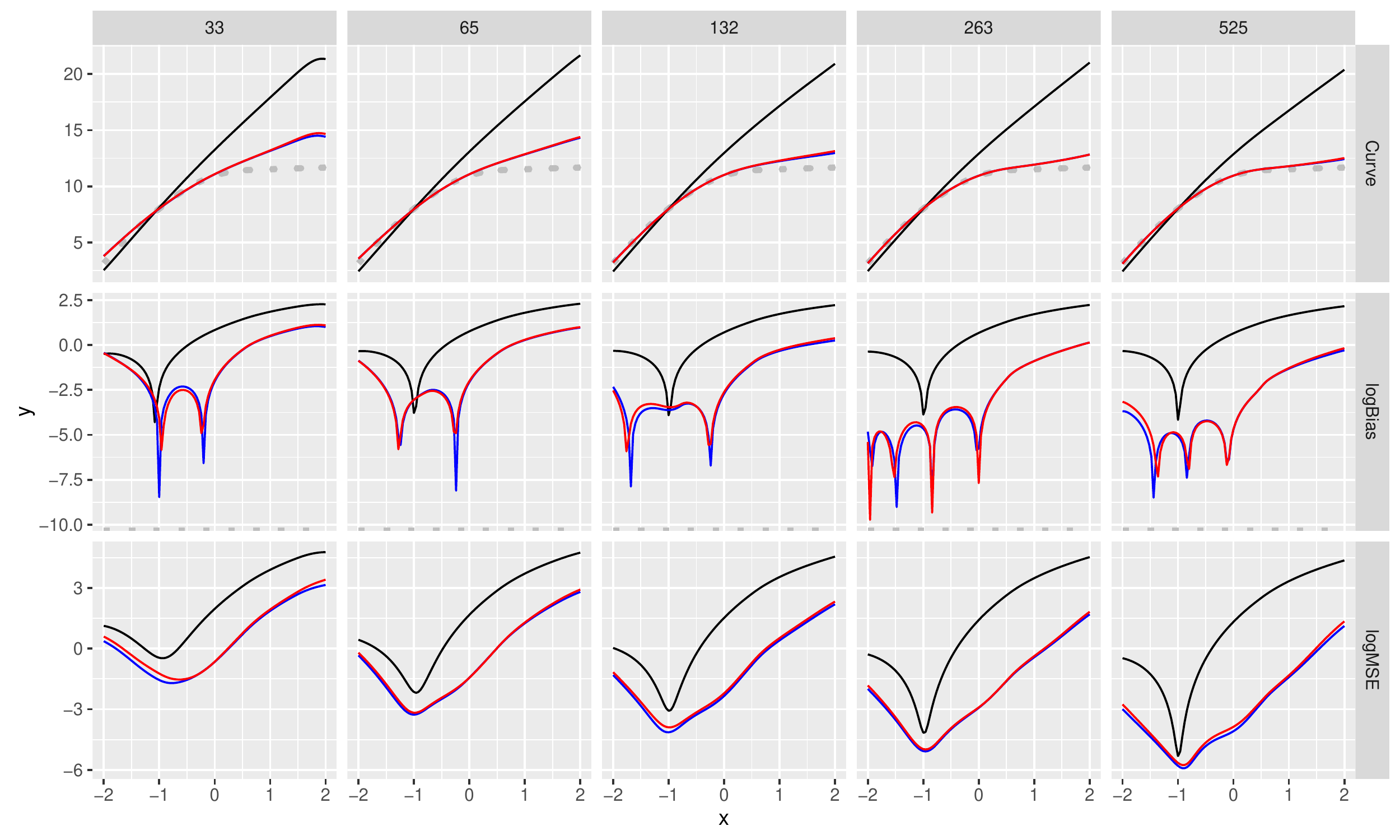}
\caption{The marginal estimate of $\mu = f(x_1)$ for P2 using the full sample under scenario~\nameref{subpop}. Compares the (true) population curve (broken grey) to the whole sample with equal weights (black), final analysis or `marginal' weights (blue), and household pairwise or `second order' weights (red). Top to bottom: estimated curve, log of absolute bias, log of mean square error. Left to right: doubling of sample size for whole sample (100 to 1600). }
\label{fig:P2allsamp}
\end{figure}
\FloatBarrier
\begin{figure}
\centering
\includegraphics[width = 0.95\textwidth,
		page = 1,clip = true, trim = 0in 0in 0in 0in]{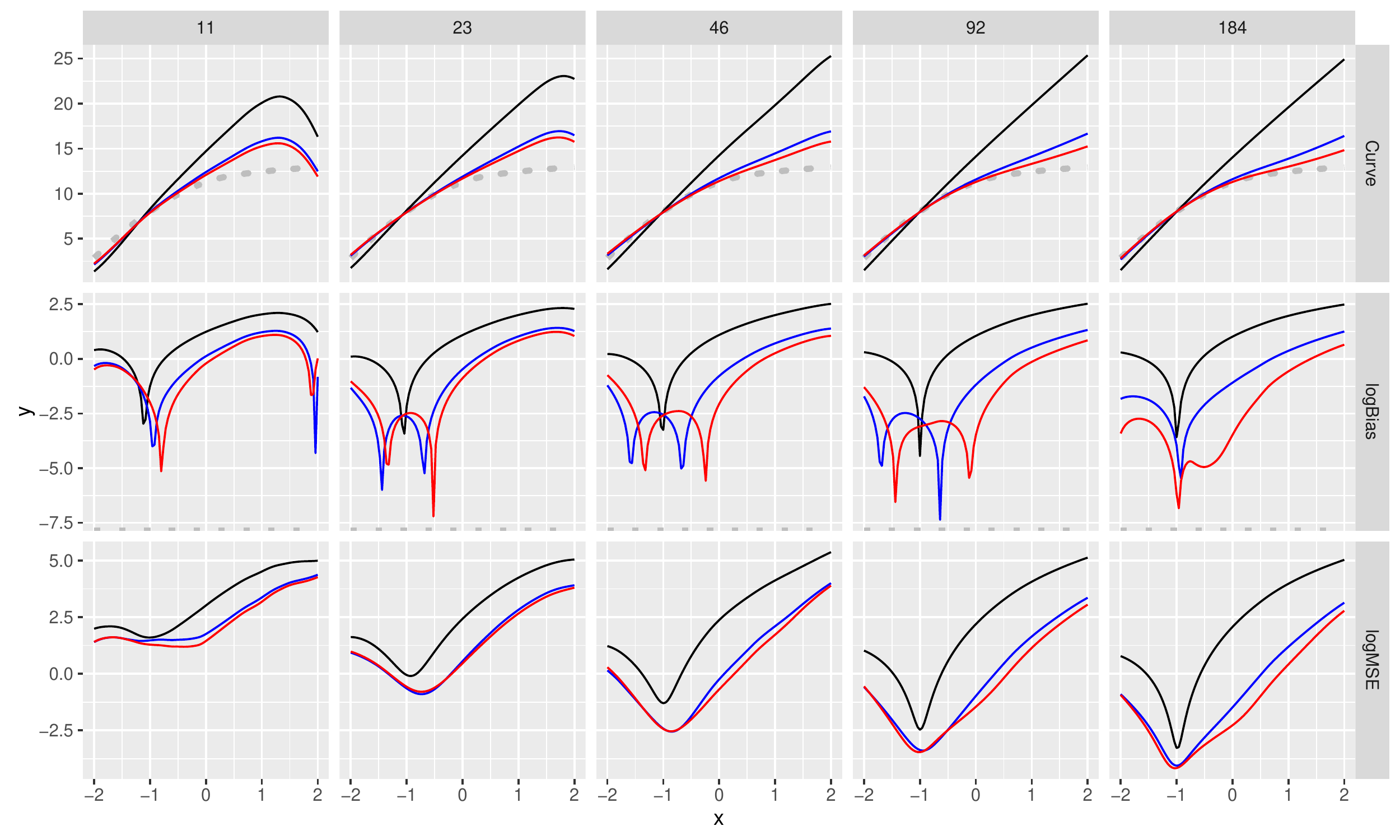}
\caption{The conditional estimate of $\mu = f(x_1)$ for P2 given observed $y \ge 10$ for P1 under scenario~\nameref{condsubpop}. Compares the (true) population curve (broken grey) to the  P1-P2 pair sample with equal weights (black), final analysis or `marginal' weights (blue), and household pairwise or `second order' weights (red). Top to bottom: estimated curve, log of absolute bias, log of mean square error. Left to right: doubling of sample size for whole sample (100 to 1600). }
\label{fig:P2condPgt}
\end{figure}
\FloatBarrier

\begin{figure}
\centering
\includegraphics[width = 0.95\textwidth,
		page = 1,clip = true, trim = 0in 0in 0in 0in]{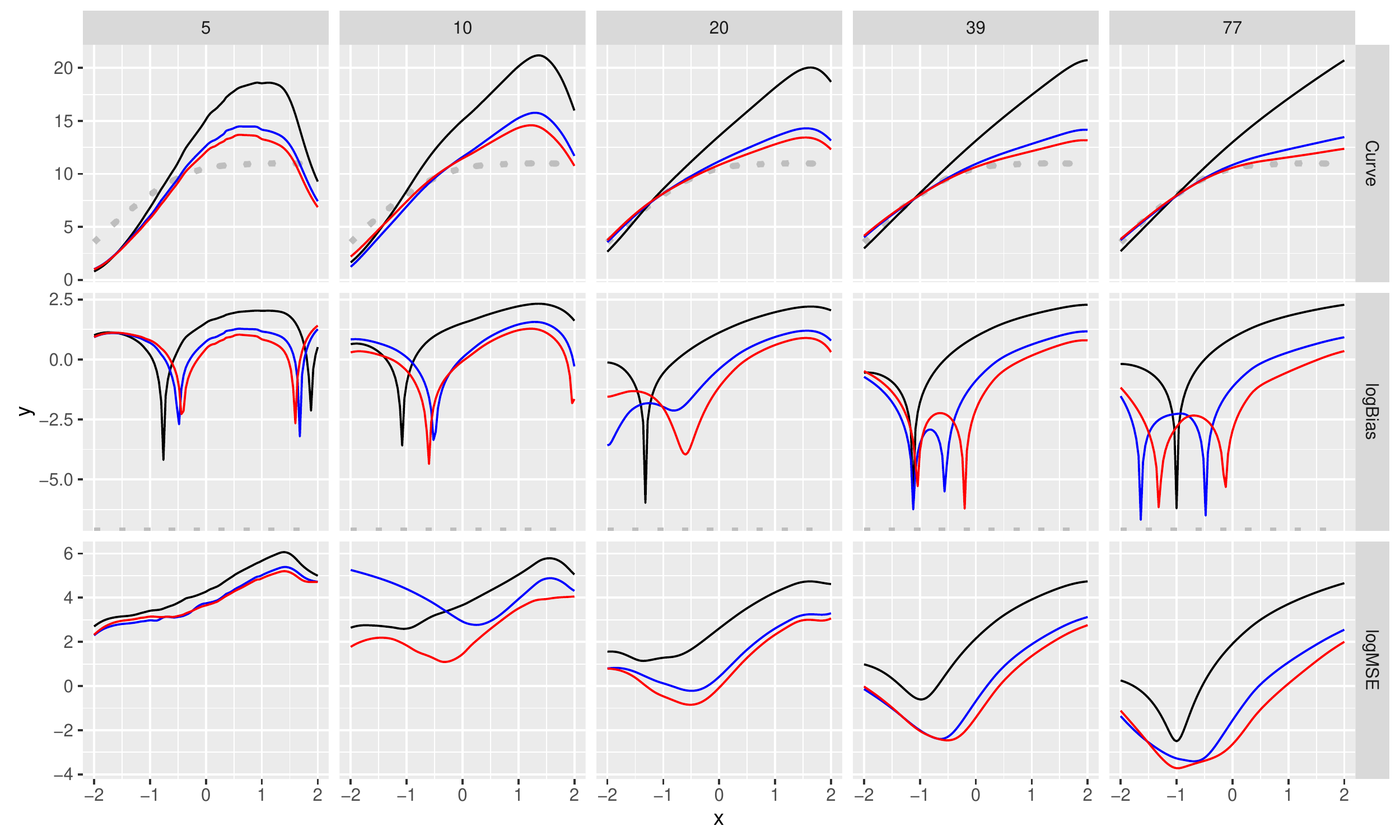}
\caption{The conditional estimate of $\mu = f(x_1)$ for P2 given observed $y < 10$ for P1 under scenario~\nameref{condsubpop}. Compares the (true) population curve (broken grey) to the  P1-P2 pair sample with equal weights (black), final analysis or `marginal' weights (blue), and household pairwise or `second order' weights (red). Top to bottom: estimated curve, log of absolute bias, log of mean square error. Left to right: doubling of sample size for whole sample (100 to 1600). }
\label{fig:P2condPlt}
\end{figure}
\FloatBarrier

\begin{figure}
\centering
\includegraphics[width = 0.95\textwidth,
		page = 1,clip = true, trim = 0in 0in 0in 0.in]{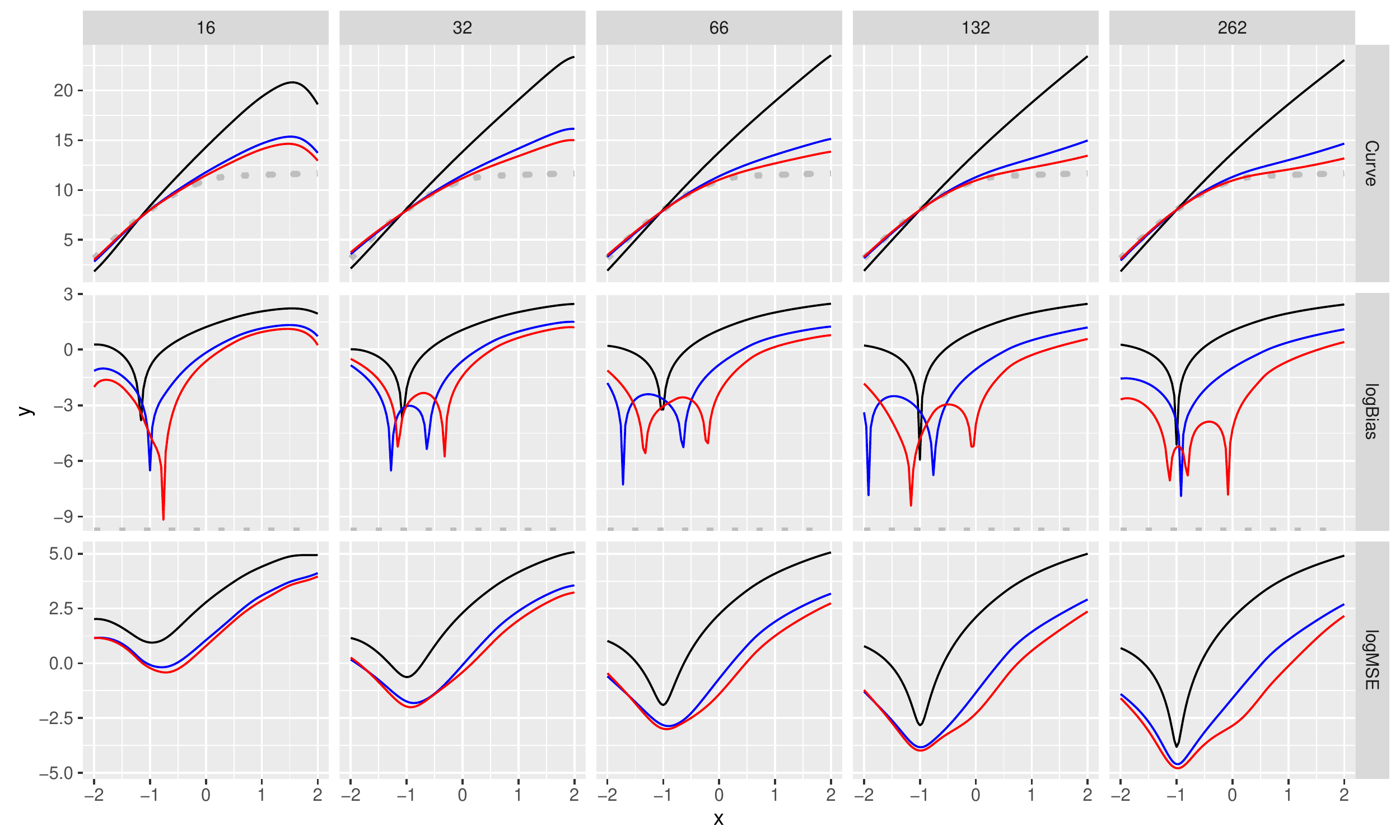}
\caption{The marginal estimate of $\mu = f(x_1)$ for P2 under scenario~\nameref{obssubpop}. Compares the (true) population curve (broken grey) to the P1-P2 pair sample with equal weights (black), final analysis or `marginal' weights (blue), and household pairwise or `second order' weights (red). Top to bottom: estimated curve, log of absolute bias, log of mean square error. Left to right: doubling of sample size for whole sample (100 to 1600). }
\label{fig:P2margP1}
\end{figure}
\FloatBarrier

\section{Application to NSDUH}
\label{sec:NSDUH}
Figure~\ref{fig:allsamp} shows the estimated relationship between the median ($q = 0.5$) number of days of alcohol use in the past month vs. age for scenario~\nameref{subpop}, that targets the sub-population of individuals who live with a spouse (which is known for the household population) and employs the entire sample of individuals ($n \approx$ 21,000) for estimation. The three weight alternatives (equal, marginal, and household pairwise) produce similar curves with overlapping 95\% intervals, though the household pairwise weighting estimates a notably less steep rise in frequency of alcohol use through young adulthood than does marginal weighting. The overlap of the estimates under both weighting schema, on the one hand, with those under equal weighting, on the other hand, may suggest that the sampling of the (unconditioned) spouse-spouse sub-population is only weakly informative if we flexibly account for age in our models (Recall that sample is allocated disproportionately across age groups and states).

Figure~\ref{fig:sppair} shows estimates that include only the sample of pairs in which both spouses were selected and responded ($n \approx$ 7,000) of scenario~\nameref{obssubpop}. Here there is less agreement among the three weighting schema, which we also realized in our simulation study. The general shape estimated under the household pairwise weights further accentuates the less steep rise in median alcohol frequency of use estimated by the pair weights using the full sample (scenario~\ref{subpop}) in Figure~\ref{fig:allsamp}.  Though the \emph{same} spouse-spouse sub-population of scenario~\nameref{subpop} is targeted in this figure as in Figure~\ref{fig:allsamp}, there are notable differences between these curves estimated under household pairwise weighting; for example, the full sample curves have a mode at an older age.

Figure~\ref{fig:spcond} shows the decomposition of the curves from Figure~\ref{fig:sppair} based on the self-reported behavior of the spouse (in their sampled response), which now targets sub-populations from scenario~\nameref{condsubpop} of individuals in spouse-spouse households, conditioned on the alcohol use of their spouses. The top set of curves corresponds to the median estimates for individuals ($n \approx$ 4,000) given that their spouse reported past month alcohol use ($y \ge 1$). The bottom set of curves represent the complement corresponding to median estimates for individuals ($n \approx$ 3,000) whose spouse reported no past month alcohol use ($y = 0$). While the bottom curves show little difference between the weighting methods and little change across age, the top set of curves for the household pairwise weights suggest a different pattern than the marginal and equal weights.  While the latter (marginal and equal weights) show a general lack of change after an initial increase in younger ages, the former (household pairwise weights) suggest a continued increase across age, perhaps with a jump around middle age; that is, alcohol consumption continues to increase with age for individuals who spouses consume alcohol at least once per month (It is important to remember that this study is cross-sectional and differences between ages may be due to cohort effects rather than a progression over time). These results provide some useful insights:
\begin{enumerate}
	\item While Figures~\ref{fig:allsamp} and \ref{fig:sppair} show a clear increase and then decrease with age, Figure \ref{fig:spcond} shows mostly constant or increasing levels among both sub-populations. This is an example of the classic issue of aggregation error (Simpson's Paradox) and can be explained by the increasing proportion of the population falling into the lower curve (spouse did not use alcohol in the past month) with increasing age.
	\item Both Figure~\ref{fig:allsamp} and Figure~\ref{fig:sppair} show different curve shapes \emph{within} each weighting scheme, even though their sub-samples both map to the same sub-population. The simulations in section~\ref{simulation} suggest that these differences in curve shapes when using the different sub-samples are likely due to the added complexity of post-stratification, non-response bias or other sources of error (not related to sampling error) present in the NSDUH sample, rather than true differences (See \citet{MRB:PersonWeight:2014} and \citet{MRB:PairWeight:2014} for more details on weighting adjustments for NSDUH). The simulation results also suggest, however, that the differences in estimated curve shapes for the sub-population \emph{between} weighting methods for both Figures \ref{fig:sppair} and \ref{fig:spcond} would be expected and may be due to the reduced bias of the household pairwise weights over the marginal and equal weights.
\end{enumerate}
\begin{figure}
\centering
\includegraphics[width = 0.95\textwidth,
		page = 1,clip = true, trim = 0in 0in 0in 0in]{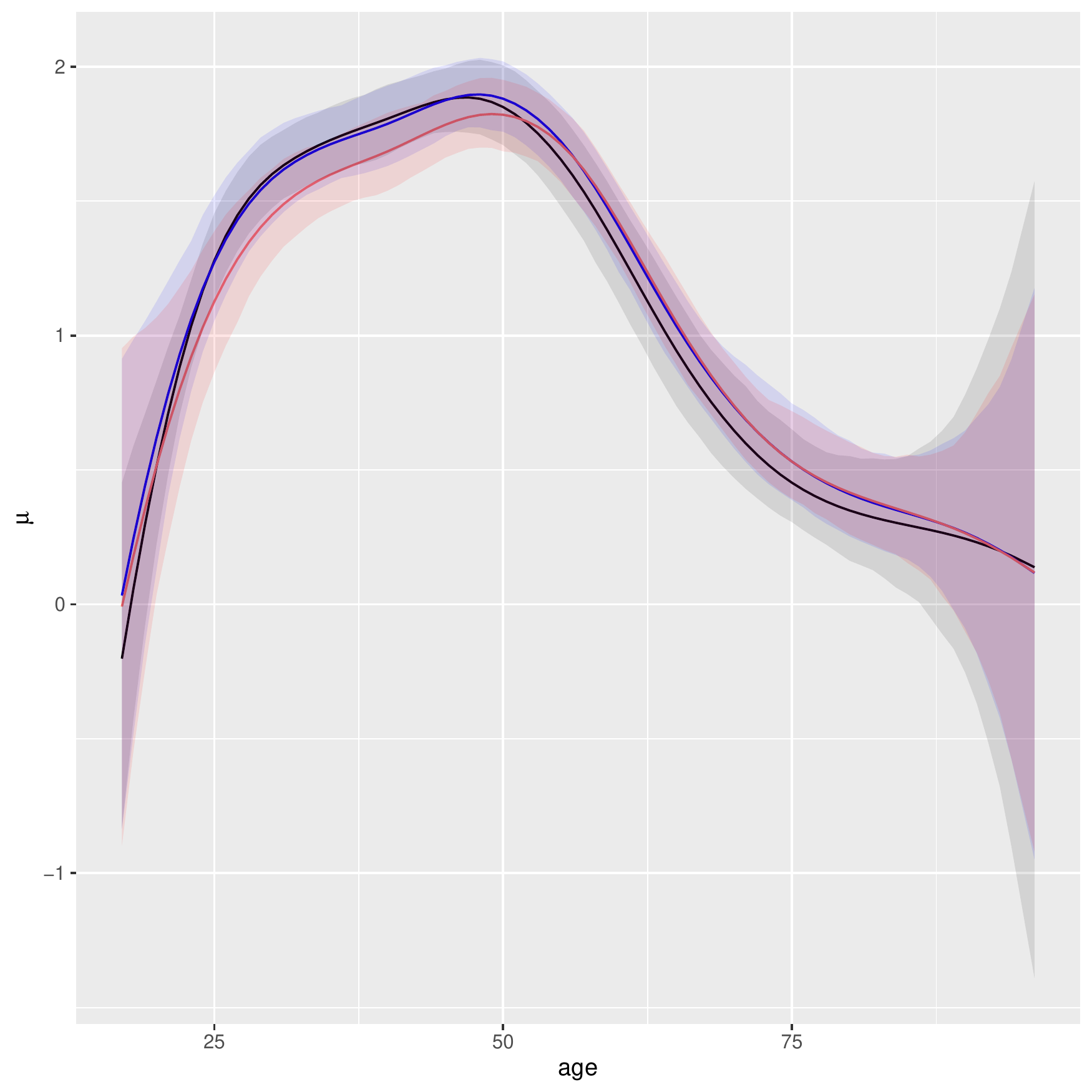}
\caption{Estimated median number of days using alcohol in the past month among individuals residing with a spouse (based on full sample) under scenario~\nameref{subpop} using equal weights (black), final analysis or `marginal' weights (blue), and household pairwise or `second order' weights (red).}
\label{fig:allsamp}
\end{figure}
\FloatBarrier

\begin{figure}
\centering
\includegraphics[width = 0.95\textwidth,
		page = 1,clip = true, trim = 0in 0in 0in 0in]{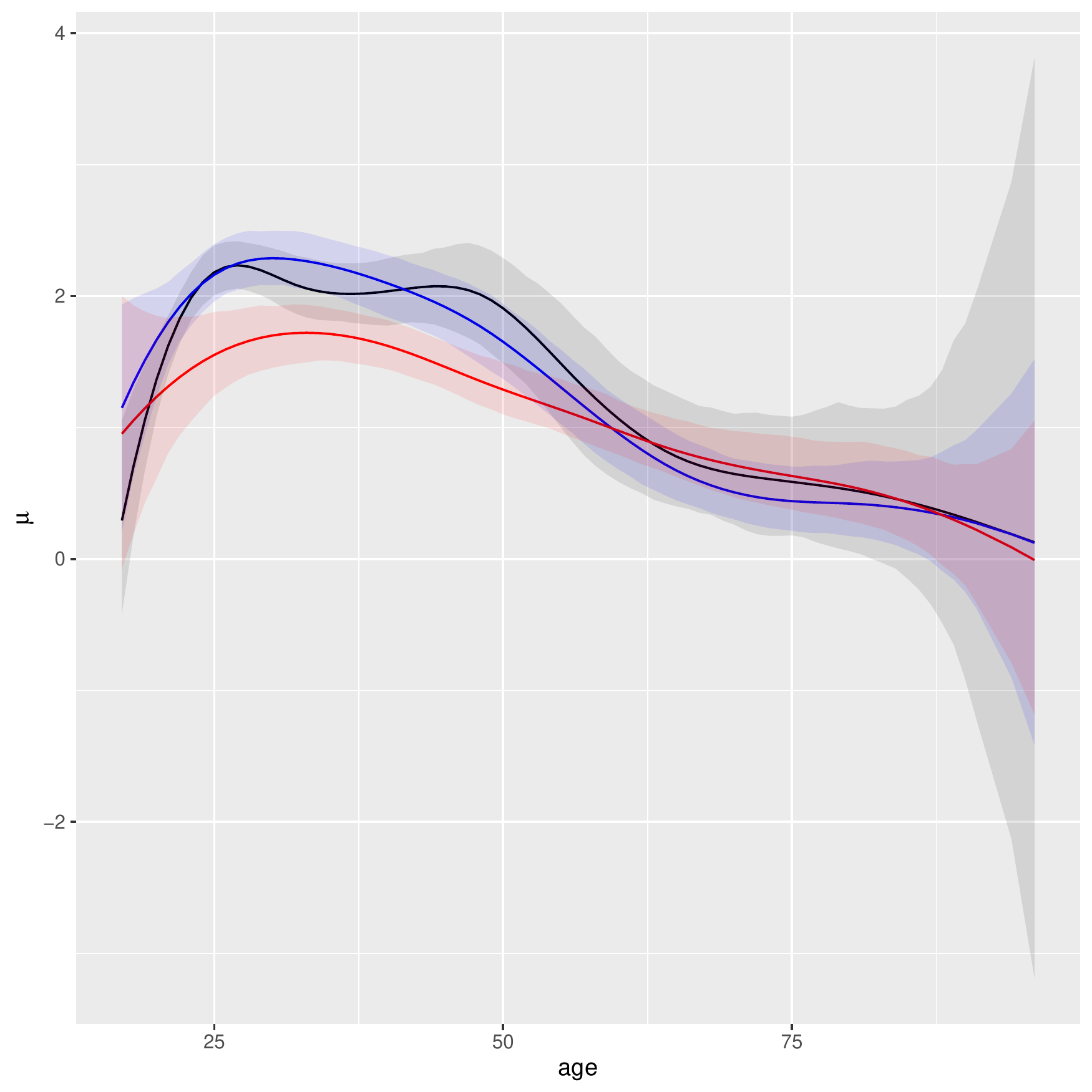}
\caption{Estimated median number of days using alcohol in the past month among individuals residing with a spouse (based on the observed pair sample) under scenario~\nameref{obssubpop} using equal weights (black), final analysis or `marginal' weights (blue), and household pairwise or `second order' weights (red).}
\label{fig:sppair}
\end{figure}
\FloatBarrier

\begin{figure}
\centering
\includegraphics[width = 0.95\textwidth,
		page = 1,clip = true, trim = 0in 0in 0in 0in]{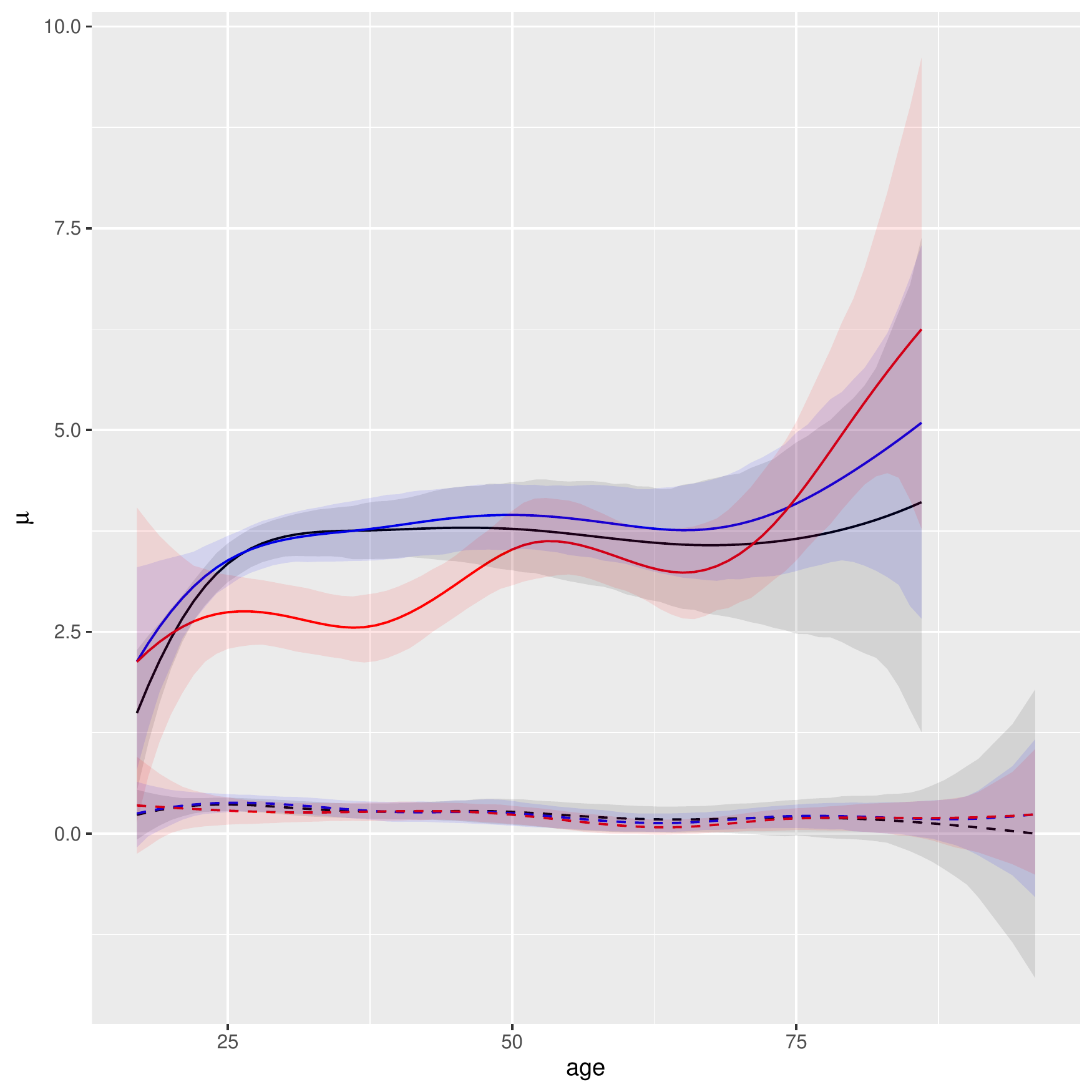}
\caption{Estimated median number of days using alcohol in the past month among individuals conditional on their spouse's past month use of alcohol (based on the observed pair sample) under scenario~\nameref{condsubpop} using equal weights (black), final analysis or `marginal' weights (blue), and household pairwise or `second order' weights (red). Solid lines indicate past month use by spouse; broken lines indicate no past month use by spouse.}
\label{fig:spcond}
\end{figure}
\FloatBarrier

\section{Conclusions}
This paper extends the previous work of \citet{2015arXiv150707050S} to include sampling designs in which the second order dependency between units does not fully attenuate. For some multi-stage surveys, such as the CE and the NSDUH, the dependence structure between most units is dominated by the nearly independent first stages of selection. This leads the `full' second order (pairwise) weights to quickly converge to the first order (marginal) weights, suggesting that this dependency is often negligible and that marginal weights are robust for inference on the general population.

It is when we (i) sub-select within the last stage of selection (e.g., household) to target inference to an associated sub-population, (ii) include only joint or conditional sample responses between members within the last stage cluster for modeling, and (iii) the sub-selection probabilities are informative (related to the outcome (e.g. size or age)) that we achieve a gain from using the second order weights.  Targeting inference to sub-populations of individuals in a household based on the behaviors of other members within the household, for example, may be of inferential interest to policy makers and researchers, but unbiased estimations for such sub-populations are not possible with first order weights.

When the inference targets the general population of individuals, however, while both first and second order weighting methods produce asymptotically unbiased inference about the population from estimation on realized samples, there may be a minor loss of efficiency for second order relative to first order weighting if the variance of the second order weights is larger than that for the first order weights. The more practical concern is the often lack of availability of second order weights.   The burden to compute and store last stage (household) pairwise weights is much reduced because it is confined to the last stage of selection, so that the weight for each individual may be constructed from a single (or, more generally, a few) pairwise term.  To the extent that stakeholders express interest to conduct inference on sub-populations of individuals conditioned on the behavior of other individuals within the last stage, this paper may help encourage statistical agencies to pursue methods to provide access to second order weights while addressing potential concerns of confidentiality.

\citet{}
\bibliography{refs_sep2017}
\bibliographystyle{agsm}

\appendix

\section{Enabling Lemmas}\label{Appendix}

\begin{lemma}\label{numerator}
Suppose conditions~\nameref{existtests} and ~\nameref{bounded} hold.  Then for every $\xi > \xi_{N_{\nu}}$, a constant, $K>0$, and any constant, $\delta > 0$,
\begin{align}
\mathbb{E}_{P_{0},P_{\nu}}\left[\mathop{\int}_{P\in\mathcal{P}\backslash\mathcal{P}_{N_{\nu}}}\mathop{\prod}_{i=1}^{N_{\nu}}
\frac{p^{\pi}}{p_{0}^{\pi}}\left(\mbf{X}_{i}\delta_{\nu i}\right)d\Pi\left(P\right)\left(1-\phi_{n_{\nu}}\right)\right] \leq \Pi\left(\mathcal{P}\backslash\mathcal{P}_{N_{\nu}}\right)& \label{outside}\\
\mathbb{E}_{P_{0},P_{\nu}}\left[\mathop{\int}_{P\in\mathcal{P}_{N_{\nu}}:d^{\pi}_{N_{\nu}}\left(P,P_{0}\right)> \delta\xi}\mathop{\prod}_{i=1}^{N_{\nu}}
\frac{p^{\pi}}{p_{0}^{\pi}}\left(\mbf{X}_{i}\delta_{\nu i}\right)d\Pi\left(P\right)\left(1-\phi_{n_{\nu}}\right)\right] &\leq \nonumber \\
2\gamma\exp\left(\frac{-K n_{\nu}\delta^{2}\xi^{2}}{\gamma}\right).&\label{inside}
\end{align}
\end{lemma}

The constant multiplier, $\gamma \geq 1$, defined in condition~\nameref{bounded}, restricts the distribution of the sampling design by bounding all marginal inclusion probabilities for population units away from $0$.  As with the main result, the upper bound is increased by $\gamma$.

\begin{proof}\label{AppNumerator}
We begin by achieving the intermediate bound of Equation $32$ in \citet{2015arXiv150707050S}, on any set $B \in \mathcal{P}$, by replacing $1/\pi_{\nu\ell}$ in Equations $29 - 31$ of \citet{2015arXiv150707050S} with $(1/(n_{\nu}-1))\mathop{\sum}_{k\neq\ell\in \bm{\delta}_{\nu}}1/\pi_{\nu \ell k}$, where $\bm{\delta}_{\nu}$ represents a particular sample of units, of fixed size $n_{\nu}$, drawn from the space of samples, $\Delta_{\nu}$.  
The upper bound result in Equation $32$ is achieved because each term, $1/\pi_{\nu \ell k}$, in the sum of $n_{\nu}-1$ terms is greater than or equal to $1$, such that the term, as a whole, is greater than or equal to $1$ for each unit, $\ell \in \bm{\delta}^{\ast}_{\nu}$, where $\bm{\delta}^{\ast}_{\nu}$ is that sample realization that maximizes $p/p_{0}\left(\mbf{X}_{\ell}\right)$ (which is less than or equal to $1$, by construction, however this can also be relaxed \citep{2016arXiv160607488S}).
The bound specified in Equation~\ref{outside} then directly follows from application of the intermediate bound on the set $\mathcal{P}\backslash\mathcal{P}_{N_{\nu}}$.

We next achieve an upper bound result for the expectation in Equation~\ref{inside} stated in Equations $36$ of \citet{2015arXiv150707050S} on models, $P$, in the slice, $\mathcal{A}^{\pi}_{r} =\{P\in \mathcal{P}_{N_{\nu}}: r\epsilon_{N_{\nu}} \leq d_{N_{\nu}}^{\pi}\left(P,P_{0}\right) \leq 2r\epsilon_{N_{\nu}} \}$ for integers, $r$, by again using the intermediate bound of Equation $32$ in \citet{2015arXiv150707050S}.  This result on a slice derives from establishing an upper bound the (pairwise) sampling weighted, pseudo Hellinger distance, as follows:
\begin{align*}
d^{\pi,2}_{N_{\nu}}\left(p_{1},p_{2}\right) &= \frac{1}{N_{\nu}}\mathop{\sum}_{i=1}^{N_{\nu}}\frac{1}{(N_{\nu}-1)}\mathop{\sum}_{k\neq i\in U_{\nu}}\frac{\delta_{\nu i}\delta_{\nu k}}{\pi_{\nu ik}}d^{2}\left(p_{1}(\mbf{X}_{i}),p_{2}(\mbf{X}_{i})\right)\\
&\leq \frac{1}{N_{\nu}}\mathop{\sum}_{i=1}^{N_{\nu}}\left[\frac{1}{(N_{\nu}-1)}\mathop{\sum}_{k\neq i\in U_{\nu}}\frac{1}{\pi_{\nu ik}}\right]d^{2}\left(p_{1}(\mbf{X}_{i}),p_{2}(\mbf{X}_{i})\right)\\
&\leq \frac{1}{N_{\nu}}\mathop{\sum}_{i=1}^{N_{\nu}}\left[\frac{1}{(N_{\nu}-1)}(N_{\nu}-1)\gamma d^{2}\left(p_{1}(\mbf{X}_{i}),p_{2}(\mbf{X}_{i})\right)\right]\\
&\leq\gamma d^{2}_{N_{\nu}}\left(P_{1},P_{2}\right),
\end{align*}
which, in turn, produces the upper bound stated in Equation $35$ of \citet{2015arXiv150707050S} that directly leads to the result in Equation $36$.  Finally, the result of Equation~\ref{inside} is directly achieved by adding up this upper bound on a slice over the countable collection of (dyadic) slices that form the set over which the integral is taken in Equation~\ref{inside}, as is outlined in the remainder of the proof in \citet{2015arXiv150707050S}.  This concludes the proof.
\end{proof}

\begin{lemma}\label{denominator}
For every $\xi > 0$ and measure $\Pi$ on the set,
\begin{equation*}
B = \left\{P:-P_{0}\log\left(\frac{p}{p_{0}}\right) \leq \xi^2, P_{0}\left(\log\frac{p}{p_{0}}\right)^{2} \leq \xi^{2}\right\}
\end{equation*}
under the conditions ~\nameref{sizespace}, ~\nameref{priortruth}, ~\nameref{bounded}, ~\nameref{factorthird}, ~\nameref{factorfourth}, we have for every $C > 0 $, $C_{3} = C_{4} + C_{5} +1$ and $N_{\nu}$ sufficiently large,
\begin{equation}\label{denomresult}
\mbox{Pr}\left\{\mathop{\int}_{P\in\mathcal{P}}\displaystyle\mathop{\prod}_{i=1}^{N_{\nu}}\frac{p^{\pi}}{p_{0}^{\pi}}
\left(\mbf{X}_{i}\delta_{\nu i}\right)d\Pi\left(P\right)\leq \exp\left[-(1+C)N_{\nu}\xi^{2}\right]\right\}
\leq \frac{\gamma+C_{3}}{C^{2} N_{\nu}\xi^{2}},
\end{equation}
where the above probability is taken with the respect to $P_{0}$ and the sampling generating distribution, $P_{\nu}$, jointly.
\end{lemma}

The bound of ``$1$" in the numerator of the result for Lemma $8.1$ of \citet{Ghosal00convergencerates}, is replaced with $\gamma + C_{3}$ for our generalization of this result in Equation~\ref{denomresult}.  The sum of positive constants, $\gamma + C_{3}$, is greater than $1$ and will be larger for sampling designs where the pairwise inclusion probabilities, $\{\pi_{\nu ij}\}$, express a relatively larger variation, which will tend to produce samples that are less representative of the underlying population.

\begin{proof}\label{AppDenominator}
The proof exactly follows that of \citet{2015arXiv150707050S} by bounding the probability expression on left-hand size of Equation~\ref{denomresult} with,
\begin{align}
&\mbox{Pr}\left\{\mathbb{G}^{\pi}_{N_{\nu}}\mathop{\int}_{P\in\mathcal{P}}\log\frac{p}{p_{0}}
d\Pi\left(P\right)\leq -\sqrt{N_{\nu}}\xi^{2}C\right\}\nonumber\\
&\leq\frac{\displaystyle\mathop{\int}_{P\in\mathcal{P}}\left[\mathbb{E}_{P_{0},P_{\nu}}
\left(\mathbb{G}^{\pi}_{N_{\nu}}\log\frac{p}{p_{0}}\right)^{2}\right]d\Pi\left(P\right)}
{N_{\nu}\xi^{4}C^{2}}\label{chebyshev:e2},
\end{align}
where we have used Chebyshev to achieve the right-hand bound of Equation~\ref{chebyshev:e2}.  We now proceed to further bound the numerator in the right-hand side of Equation~\ref{chebyshev:e2}, which will result in the expression on the right-hand side of Equation~\ref{denomresult}. The expectation inside the square brackets on the right-hand side of Equation~\ref{chebyshev:e2} is taken with respect to the joint distribution of population generation and the taking of a sample.  In the sequel, define $\mathcal{A}_{\nu} = \sigma\left(\mbf{X}_{1},\ldots,\mbf{X}_{N_{\nu}}\right)$ as the sigma field of information potentially available for the $N_{\nu}$ units in population, $U_{\nu}$.

\begin{subequations}
\begin{align}
&\mathbb{E}_{P_{0},P_{\nu}}\left[\mathbb{G}^{\pi}_{N_{\nu}}\log\frac{p}{p_{0}}\right]^{2}\\
&= \mathbb{E}_{P_{0},P_{\nu}}\left[\sqrt{N_{\nu}}\left(\mathbb{P}^{\pi}_{N_{\nu}} - \mathbb{P}_{N_{\nu}}\right)\log\frac{p}{p_{0}}-\sqrt{N_{\nu}}\left(\mathbb{P}_{0} - \mathbb{P}_{N_{\nu}}\right)\log\frac{p}{p_{0}}\right]^{2}\\
&= \mathbb{E}_{P_{0},P_{\nu}}\left[\sqrt{N_{\nu}}\left(\mathbb{P}^{\pi}_{N_{\nu}} - \mathbb{P}_{N_{\nu}}\right)\log\frac{p}{p_{0}}-\sqrt{N_{\nu}} \mathbb{G}_{N_{\nu}}\log\frac{p}{p_{0}}\right]^{2}\\
&\leq N_{\nu}\mathbb{E}_{P_{0},P_{\nu}}\left[\left(\mathbb{P}^{\pi}_{N_{\nu}} - \mathbb{P}_{N_{\nu}}\right)\log\frac{p}{p_{0}}\right]^{2} + \mathbb{E}_{P_{0}}\left[\mathbb{G}_{N_{\nu}}\log\frac{p}{p_{0}}\right]^{2}\label{gbound}\\
&\leq N_{\nu}\mathbb{E}_{P_{0},P_{\nu}}\left[\left(\mathbb{P}^{\pi}_{N_{\nu}} - \mathbb{P}_{N_{\nu}}\right)\log\frac{p}{p_{0}}\right]^{2} +~ \xi^{2},
\end{align}
\end{subequations}
where the bound of the expectation of the centered empirical process approximation over the units in the population taken with respect to the population generating distribution, included in the second term in Equation~\ref{gbound}, is shown to be bounded from above (for any constant, $C > 0$) under Lemma $B.2$ of \citet{2015arXiv150707050S} by replacing $(\gamma + C_{3})$ with ``$1$" in the bound $\xi^{2}(\gamma + C_{3})$ because we draw a finite population from $P_{0}$ and do not take a further informative sample under $P_{\nu}$.

We now proceed to further simplify the bound in the first term of Equation~\ref{gbound}.

\begin{subequations}
\begin{align}
&N_{\nu}\mathbb{E}_{P_{0},P_{\nu}}\left[\left(\mathbb{P}^{\pi}_{N_{\nu}} - \mathbb{P}_{N_{\nu}}\right)\log\frac{p}{p_{0}}\right]^{2}\\
&=N_{\nu}\mathbb{E}_{P_{0},P_{\nu}}\left[\frac{1}{N_{\nu}}\mathop{\sum}_{i=1}^{N_{\nu}}\left(\{\frac{1}{(N_{\nu}-1)}\mathop{\sum}_{k\neq i \in U_{\nu}}\frac{\delta_{\nu i}\delta_{\nu k}}{\pi_{\nu ik}}\} - 1\right)\log\frac{p}{p_{0}}\left(\mbf{X}_{i}\right)\right]^{2}\\
&=\frac{1}{N_{\nu}}\mathop{\sum}_{i,j\in U_{\nu}}\mathbb{E}_{P_{0},P_{\nu}}\left[\left(\{\frac{1}{(N_{\nu}-1)}\mathop{\sum}_{k\neq i \in U_{\nu}}\frac{\delta_{\nu i}\delta_{\nu k}}{\pi_{\nu ik}}\} - 1\right)\times\right.\nonumber\\
&\left.{}\left(\{\frac{1}{(N_{\nu}-1)}\mathop{\sum}_{\ell\neq j \in U_{\nu}}\frac{\delta_{\nu j}\delta_{\nu \ell}}{\pi_{\nu j\ell}}\} - 1\right)\log\frac{p}{p_{0}}\left(\mbf{X}_{i}\right)\frac{p}{p_{0}}\left(\mbf{X}_{j}\right)\vphantom{\frac{1}{(N_{\nu}-1)}}\right]^{2}\\
&=\frac{1}{N_{\nu}}\mathop{\sum}_{i \neq j\in U_{\nu}}\mathbb{E}_{P_{0}}\left[\mathbb{E}_{P_{\nu}}\left\{\frac{1}{(N_{\nu}-1)^{2}}\mathop{\sum}_{k\neq i,\ell\neq j \in U_{\nu}}\frac{\delta_{\nu i}\delta_{\nu k}\delta_{\nu j}\delta_{\nu \ell}}{\pi_{\nu ik}\pi_{\nu j\ell}}\right.\right.\nonumber\\
&\left.{}+\frac{1}{(N_{\nu}-1)^{2}}\mathop{\sum}_{k\neq i,k\neq j \in U_{\nu}}\frac{\delta_{\nu i}\delta_{\nu j}\delta_{\nu k}}{\pi_{\nu ik}\pi_{\nu jk}}\right.\nonumber\\
&\left.{}-\frac{1}{(N_{\nu}-1)}\mathop{\sum}_{k\neq i \in U_{\nu}}\frac{\delta_{\nu i}\delta_{\nu k}}{\pi_{\nu ik}}\right.\nonumber\\
&\left.\left. -\frac{1}{(N_{\nu}-1)}\mathop{\sum}_{\ell\neq j \in U_{\nu}}\frac{\delta_{\nu j}\delta_{\nu \ell}}{\pi_{\nu j\ell}} + 1 \middle\vert \mathcal{A}_{\nu}\vphantom{\frac{1}{(N_{\nu}-1)^{2}}}\right\}\left(\log\frac{p}{p_{0}}\left(\mbf{X}_{i}\right)\frac{p}{p_{0}}\left(\mbf{X}_{j}\right)\right)
\vphantom{\mathbb{E}_{P_{\nu}}}\right]\nonumber\\
&+ \frac{1}{N_{\nu}}\mathop{\sum}_{i=j\in U_{\nu}}\mathbb{E}_{P_{0}}\left[\mathbb{E}_{P_{\nu}}\left\{\frac{1}{(N_{\nu}-1)^{2}}\mathop{\sum}_{k\neq\ell\neq i \in U_{\nu}}\frac{\delta_{\nu i}\delta_{\nu k}\delta_{\nu \ell}}{\pi_{\nu ik}\pi_{\nu i\ell}} \right.\right.\nonumber\\
&\left.\left.{}-\frac{2}{(N_{\nu}-1)}\mathop{\sum}_{k\neq i \in U_{\nu}}\frac{\delta_{\nu i}\delta_{\nu k}}{\pi_{\nu ik}} + 1\middle\vert\mathcal{A}_{\nu}\right\}\left(\log\frac{p}{p_{0}}\left(\mbf{X}_{i}\right)\right)^{2}\vphantom{\frac{1}{(N_{\nu}-1)^{2}}}\right]\nonumber\\
&+ \frac{1}{N_{\nu}}\mathop{\sum}_{i=j\in U_{\nu}}\mathbb{E}_{P_{0}}\left[\mathbb{E}_{P_{\nu}}\left\{\frac{1}{(N_{\nu}-1)^{2}}\mathop{\sum}_{k\neq i \in U_{\nu}}\frac{\delta_{\nu i}\delta_{\nu k}}{\pi_{\nu ik}^{2}} \middle\vert\mathcal{A}_{\nu}\right\}\left(\log\frac{p}{p_{0}}\left(\mbf{X}_{i}\right)\right)^{2}\vphantom{\frac{1}{(N_{\nu}-1)^{2}}}\right]\\
&= \frac{1}{N_{\nu}}\mathop{\sum}_{i \neq j\in U_{\nu}} \mathbb{E}_{P_{0}}\left[
	\left\{\frac{1}{(N_{\nu}-1)^{2}}\mathop{\sum}_{k\neq i, \ell\neq j \in U_{\nu}}\frac{\pi_{\nu ikj\ell}}{\pi_{\nu ik}\pi_{\nu j\ell}} - 1  + \frac{1}{(N_{\nu}-1)^{2}}\mathop{\sum}_{k\neq i, k\neq j \in U_{\nu}}\frac{\pi_{\nu ijk}}{\pi_{\nu ik}\pi_{\nu jk}}\right\}  \right.\nonumber\\
&\left. \left( \log\frac{p}{p_{0}}\left(\mbf{X}_{i}\right)\frac{p}{p_{0}}\left(\mbf{X}_{j}\right)\right) \right]\nonumber\\
&+ \frac{1}{N_{\nu}}\mathop{\sum}_{i = j\in U_{\nu}} \mathbb{E}_{P_{0}}\left[
\left\{\frac{1}{(N_{\nu}-1)^{2}}\mathop{\sum}_{k\neq \ell \neq i \in U_{\nu}}\frac{\pi_{\nu ik\ell}}{\pi_{\nu ik}\pi_{\nu i\ell}} - 1\right\}\left(\log\frac{p}{p_{0}}\left(\mbf{X}_{i}\right)\right)^{2} \right]\nonumber\\
&+ \frac{1}{N_{\nu}}\mathop{\sum}_{i = j\in U_{\nu}} \mathbb{E}_{P_{0}}\left[
\left\{\frac{1}{(N_{\nu}-1)^{2}}\mathop{\sum}_{k\neq i \in U_{\nu}}\frac{1}{\pi_{\nu ik}} \right\}\left(\log\frac{p}{p_{0}}\left(\mbf{X}_{i}\right)\right)^{2} \right]\\
&\leq (N_{\nu}-1)\mathop{\sup}_{\nu}\mathop{\max}_{i,j,k,\ell: i \neq j, k\neq i, \ell\neq j\in U_{\nu}}\left\vert\frac{\pi_{\nu ikj\ell}}{\pi_{\nu ik}\pi_{\nu j\ell}} - 1\right\vert \left\{\mathbb{E}_{P_{0}}\log\frac{p}{p_{0}}\left(\mbf{X}_{i}\right)\frac{p}{p_{0}}\left(\mbf{X}_{j}\right)\right\} \nonumber\\
&+ \mathop{\sup}_{\nu}\mathop{\max}_{i,k,\ell: k\neq\ell\neq i\in U_{\nu}}\left\vert\frac{\pi_{\nu k\ell i}}{\pi_{\nu k i}\pi_{\nu \ell i}}\right\vert \left\{\mathbb{E}_{P_{0}}\log\frac{p}{p_{0}}\left(\mbf{X}_{i}\right)^{2}\right\} \nonumber\\
&+ \mathop{\sup}_{\nu}\left[\frac{1}{(N_{\nu}-1)}\frac{1}{\displaystyle\mathop{\min}_{i,k:k\neq i\in U_{\nu}}\vert\pi_{\nu ik}\vert}\right] \left\{\mathbb{E}_{P_{0}}\log\frac{p}{p_{0}}\left(\mbf{X}_{i}\right)^{2}\right\}
\leq (C_{4}+C_{5}+\gamma)\xi^{2},
\end{align}
\end{subequations}
for sufficiently large $N_{\nu}$, where we have applied the condition for $P\in B$ in each of the three terms in the last inequality and conditions ~\nameref{factorfourth}, ~\nameref{factorthird} and  ~\nameref{bounded} for each term in the last inequality, from left-to-right.  We additionally note that $\pi_{\nu ik\ell} = \pi_{\nu ik}$ when $\ell = k,~\ell,k \in U_{\nu}$ and denote $\pi_{\nu k\vert i} := \mbox{Pr}\left(\delta_{\nu k} = 1\vert \delta_{\nu i} = 1\right)$.

We may finally bound the expectation on the right-hand size of Equation~\ref{chebyshev:e2},
\begin{multline}
\mathbb{E}_{P_{0},P_{\nu}}\left[\mathbb{G}^{\pi}_{N_{\nu}}\log\frac{p}{p_{0}}\right]^{2} \leq (C_{4}+C_{5}+\gamma)\xi^{2} \\+ \xi^2 \leq (C_{4}+C_{5}+\gamma+ 1)\xi^{2} \leq (C_{3} +\gamma)\xi^2,
\end{multline}
for $N_{\nu}$ sufficiently large, where we set $C_{3} := C_{4}+C_{5}+1$.  This concludes the proof.

\end{proof}

\section{Calculation of Second Order Weights}\label{sec:weightcalc}
The calculation of second order weights can be motivated in different ways:
\begin{enumerate}
	\item ``Full'' second order weights can be constructed by populating the full $n \choose 2$ matrix of second order inclusion probabilities across the entire sample $\pi_{i,i'}$, taking their inverse $w^{(2)}_{i,i'} = 1/\pi_{i,i'}$, and then summing for each individual record and normalizing by the number of pairs in the target population $w^{(2f)}_{i} = \sum_{i'} w^{(2)}_{i,i'}/(N -1)$ where $N$ is the population size. Because there are three stages, the joint weight component $w^{(2)}_{i,i'}$ is calculated differently depending on whether individuals $i$ and $i'$ are in the same HH, different HHs but same PSU, or different PSUs:
	\begin{itemize}
	\item Case 1: $i$ and $i'$ are in different PSUs. Since we treat selection of PSUs as independent and we select HHs and persons separately across PSUs, the second order weight component is simply the product of the first order weights: $w^{(2)}_{i,i'} = w_{i}^{(1)} w_{i'}^{(1)}$.
	\item Case 2: $i$ and $i'$ are in the same PSU but different HHs: Since we treat HHs as being conditionally independent given the selection of the PSUs, the second order weight term has a common PSU term and distinct HH and individual terms: $w^{(2)}_{i,i'} = w_1^{k} w_2^{j|k}w_3^{i|jk} w_2^{j'|k}w_3^{i'|j'k}$
	\item Case 3: $i$ and $i'$ are in the same HH and thus selected as a pair: The second order weight term has a common PSU term and common HH  term and common person-pair term: $w^{(2)}_{i,i'} = w_1^{k} w_2^{j|k}w_3^{i,i'|jk}$
	\end{itemize}
	\item `Last stage' or `pair' weights can be constructed by assuming the HH is the unit of analysis. This means that first order weight components provide a consistent estimate at the HH level and any sampling dependence across households is negligible. Because responses are not available for all persons within a household, we then incorporate the joint sampling dependence of each pair selected. Instead of summing over all $i,i'$ pairs in the sample, we only sum over the $i,i'$ pairs in each HH, which in this example is a single pair. We use the household size $(N_{p_j})$ to normalize by the number of pairs within each household because each roster of the HH is treated as a population:
	$w^{(2p)}_{i} = w^{(2)}_{i,i'}/(N_{p_j} -1) = w_1^{k} w_2^{j|k}w_3^{i,i'|jk} /(N_{p_j} -1)$. This set of weights is appealing for several reasons. It is most likely that dependence in the response $y$ is strongest within a HH rather than between HHs and PSUs. It is also possible to measure each $N_p$ during data collection. It usually needs to be known before sampling can occur. Whereas the general $N$ for a particular domain may not be available and must be estimated by summing the first order weights.
	\item `Stagewise' second order weights can be constructed by assuming that each of the three stages is a conditionally independent sampling with conditional population frames of PSUs, HHs, and individuals. So any second order dependence would only need to be captured within each stage. Then the first order weights for each stage would be replaced by the scaled sum of the second order weights for that stage. For example the stage 1 weights for PSUs would be $w_{1^{'}}^{k} = \sum_{k^{'}} w_1^{k,k'}/(N_S -1)$ with $ w_1^{k,k'} = 1/\pi_{k,k^{'}}$ and $N_S$ is the number of PSUs in the population. Similarly $w_{2^{'}}^{j|k} = \sum_{j^{'}} w_2^{j,j^{'}|k}/(N_{h_k} -1)$ with $w_2^{j,j^{'}|k} = 1/\pi_{j,j'|k}$ and $N_{h_k}$ is the number of HHs in PSU $k$. The within household stage weights would be $w_{3^{'}}^{i|jk} =w_3^{i,i^{'}|jk} /(N_{p_j} -1)$. Then the final stagewise second order weight would be
\begin{multline*}
w^{(2s)}_{i} = w_{1^{'}}^{k} w_{2^{'}}^{j|k} w_{3^{'}}^{i|jk} = \\
	\left(\sum_{k^{'}} w_1^{k,k^{'}}/(N_S -1)\right) \times
	\left( \sum_{j^{'}} w_2^{j,j^{'}|k}/(N_{h_k} -1) \right) \times
	 \left( w_3^{i,i^{'}|jk} /(N_{p_j} -1)\right)
\end{multline*}
\end{enumerate}

For moderate to large samples, particularly for multi-stage designs with low dependence at the first levels of sampling, the four weights presented here ($w_{i}^{(1)}$, $w^{(2f)}_{i}$, $w^{(2p)}_{i}$, $w^{(2s)}_{i}$) effectively become only two distinct sets of weights ($w_{i}^{(1)}$ and $w^{(2p)}_{i}$). For full second order weights $w^{(2f)}_{i}$, the sum is dominated by Case 1. Then the sum  $w^{(2f)}_{i} = \sum_{i'} w^{(2)}_{i,i'}/(N -1) \rightarrow w_{i}^{(1)} (\sum_{i^{'}} w_{i'}^{(1)})/(N-1)$. The sum of the first order weights $\left( \sum_{i} w_i = \sum_{i'} w_{i'} + w_i \right)$ converges to $N$ so   $w^{(2f)}_{i} \approx \left( \frac{N-w_{i}^{(1)}}{N-1} \right) w_{i}^{(1)} \approx w_{i}^{(1)}$. For stagewise weights $w^{(2s)}_{i}$, the first and second stage components from the PPS design are typically assumed to be the products of independent samples. By a similar argument, for moderate sample sizes  $w_{1'}^{k}$ and $w_{2'}^{j|k}$ is effectively the same as $w_1^{k}$ and $w_2^{j|k}$. The third or last stage weights are still distinct, so the stagewise weights $w^{(2s)}_{i}$ are very similar to the last stage pair weights $w^{(2p)}_{i}$.

\end{document}